\documentclass[journal]{IEEEtran}
\ifCLASSINFOpdf
\else
\fi
\usepackage{graphicx}
\usepackage{amssymb}
\usepackage{amsbsy}
\usepackage{amsmath}
\usepackage{amsthm} % This is only for theorem and proofs.
\usepackage{comment}
\usepackage{caption}
\usepackage{subcaption}
\usepackage{color}
\usepackage{pgf}
\usepackage{float}
\newtheorem{thm}{Theorem}
\newtheorem{cor}{Corollary}
\newtheorem{lem}{Lemma}
\newtheorem{remark}{Remark}
\newcounter{MYtempeqncnt}
\begin{document}
%
% paper title
% can use linebreaks \\ within to get better formatting as desired
\title{Downlink and Uplink Decoupling in Two-Tier Heterogeneous Networks with Multi-Antenna Base Stations}
%\author{Mudasar Bacha,}
\author{\IEEEauthorblockN{Mudasar Bacha, Yueping Wu and Bruno Clerckx}\thanks{M. Bacha and B. Clerckx is with the Communication and Signal Processing Group, Department
of Electrical and Electronic Engineering, Imperial College London, London
SW7 2AZ, U.K. (e-mail: {m.bacha13, b.clerckx}@imperial.ac.uk). This work has been partially supported by the EPSRC of UK, under grant EP/N015312/1.} \thanks{Y. Wu was with the Department of Electrical and Electronic Engineering, Imperial College London, London SW7 2AZ, U.K., and is now with the Hong Kong Applied Science and Technology Research Institute (ASTRI) (e-mail: veronicawu@astri.org). } }
%\thanks{Y. Wu was with the Department of Electrical and Electronic Engineering,
%Imperial College London, London SW7 2AZ, U.K., and is now with the Department
%of Electrical and Electronic Engineering, The University of Hong Kong,
%Hong Kong (e-mail: eewyp@hku.hk).}
%\thanks{B. Clerckx is with the Department of Electrical and Electronic Engineering,
%Imperial College London, London SW7 2AZ, U.K., and also with the School
%of Electrical Engineering, Korea University, Seoul 136-713, Korea (e-mail:
%b.clerckx@imperial.ac.uk).}
%\IEEEauthorblockA{Communications and Signal Processing Group\\Department of Electrical and Electronic Engineering\\
%Imperial College London}}
%.Department of ..., ... Institute of ..., City,
%Country, e-mail: \protect\href{http://xxx@xxx.xxx}{m.bacha13@imperial.ac.uk}.%}%
\maketitle
\begin{abstract}
%\boldmath
In order to improve the uplink performance of future cellular networks, the idea to decouple the downlink (DL) and uplink (UL) association has recently been shown to provide significant gain in terms of both coverage and rate performance. However, all the work is limited to SISO network. Therefore, to study the gain provided by the DL and UL decoupling in multi-antenna base stations (BSs) setup, we study a two tier heterogeneous network consisting of multi-antenna BSs, and single antenna user equipments (UEs). We use maximal ratio combining (MRC) as a linear receiver at the BSs and using tools from stochastic geometry, we derive tractable expressions for both signal to interference ratio (SIR) coverage probability and rate coverage probability. We observe that as the disparity in the beamforming gain of both tiers increases, the gain in term of SIR coverage probability provided by the decoupled association over non-decoupled association decreases. We further observe that when there is asymmetry in the number of antennas of both tier, then we need further biasing towards femto-tier on the top of decoupled association to balance the load and get optimal rate coverage probability.
\end{abstract}
\IEEEpeerreviewmaketitle
\section{Introduction}
The demand for high data rates is ever-growing and it is projected that over the next decade a factor of a thousand times increase in wireless network capacity will be required \cite{qualcom}. In order to meet this challenge, a massive densification of the current wireless networks characterized by the dense deployment of low power and low cost small cell is required, which will convert the existing single-tier homogeneous networks into multi-tier heterogeneous networks (HetNets) \cite{geff_femto}. HetNets that consist of different types of base stations (Macro, Micro, Pico and Femto)  can not be operated in the same way as a single-tier homogeneous network (consisting of Macro Base stations only) and need some fundamental changes in the design and deployment to meet the high data rate demand.

Cellular networks have been designed mainly for downlink (DL) because initially the traffic was asymmetric (mostly in the downlink direction). However, with the increase in real-time applications, online social-networking, and video-calling the traffic in UL has greatly increased, which necessitates the need for the uplink (UL) optimization. In current cellular networks, cell association is based on downlink average received power, which is viable for homogeneous networks where the transmit power of all the base stations (BSs) is the same. However, in heterogeneous networks there is a big disparity in the transmit power of different BSs, and the association scheme based on downlink received power is highly inefficient, therefore, the idea of downlink and uplink decoupling (DUDe) has been proposed for 5G in \cite{geff_mag_7}-\cite{jeff_mag_dude}. %\cite{5disruptive}. 
\subsection{Related Work}
A simulation based study has been performed on two-tier live network where the UL association is based on minimum path-loss while the DL association is based on DL received power \cite{mischa1}.  This kind of association divides the users into three groups: users attached to macro base station (MBS) both in the DL and UL, users attached to femto base station (FBS) both in the DL and UL, and users attached to MBS in the DL and FBS in the UL.  The authors in \cite{mischa1} showed that the gain in UL throughput is quite high when the UL association is based on minimum path-loss. The gain comes from those users which are connected to MBS in the DL and FBS in the UL because they have better channel to the femto-cell and they create less interference to the macro-cell. A network consisting of macro-tier and femto-tier is studied using tools from stochastic geometry in \cite{katerina}, where the throughput gain due to decoupling has been shown. In \cite{mischa2}, the analytical results obtained from stochastic geometry-based model have been compared with the results obtained from simulation in \cite{mischa1}, and they found that both of them match with each other.  They also found that the association probability mainly depends on the density of the deployment and not on the process used to generate the deployment geometry. It has been shown in \cite{New_DUDe} that DUDe provides gain in term of system rate, spectrum efficiency, and energy efficiency. A joint study of DL and UL for $k$ tier SISO network has been performed in \cite{Geff_DUDE}, while considering a weighted path-loss association and UL power control.

Stochastic geometry has emerged as a powerful tool for the analysis of cellular networks after the seminal work of \cite{Geff_DL}. It has been shown that stochastic geometry-based models are equally accurate as  grid based models. In addition, they provide more tractability and their accuracy becomes better as the heterogeneity of the network increases. Most of the work, which considered stochastic geometry-based model mainly studied the DL performance of the HetNets. For instance, single input single output (SISO) HetNets have been studied in \cite{Han_Geff}, and \cite{offloading}, MIMO HetNets in \cite{yueping}, \cite{yueping2}, \cite{MIMO_hetnets1}, \cite{MIMO_hetnets_LB}, and \cite{MIMO_hetnets_assoc}. A complete survey can be found in \cite{hesham_survey} and the references therein.  However, only limited work has been carried out in UL because it is more involved due to UL power control and correlation among interferers. The UL power control is required because an interfering user may be closer to the BS than the scheduled user, which creates an additional source of randomness in the UL modeling. The correlation among interfering users comes due to the orthogonal channels assignment within a cell, which prohibits the use of the same channel in a cell. In other words, there is only one UE randomly located  within the coverage region of the BS, which transmits in a given resource block. Therefore, the interference does not originate from the the PPP distributed UEs but instead from Voronoi perturbed lattice process \cite{Geff_DUDE}. The exact interference characterization for which is not available \cite{Geff_DUDE}, and thus makes the UL analysis even more involved. 

An uplink model for the single tier network has been derived in \cite{UL_tier1}, which uses fractional power control (FPC) in the UL. A multi-tier UL performance has been studied in \cite{hisham} and \cite{Geff_DUDE}, where each tier differs only in terms of density, cutoff threshold, and transmit power.  In \cite{hisham} a truncated channel power inversion is used due to which mobile users suffer from truncation outage in addition to SINR outage. The performance gain of DUDe is only studied for SISO network and there is no work which studies the decoupled association in the MIMO network\footnote{\cite{marco} studies the UL performance in multi antennas BSs network, which was not available online at the initial submission of this paper. However, our analysis approach is significantly different than \cite{marco}. We explicitly take into account the beamforming gain in the cell association and use Fa$\grave{\text{a}}$ di Bruno's formula \cite{Faa_de_brunos} to find the high order derivative of the Laplace transform of the interference, whereas \cite{marco} does not consider beamforming gain in the cell association and  use Gil-Pelaez inversion theorem to avoid finding the higher order derivative.}. Therefore, in this work we consider multi-antenna BSs and we also consider UL biasing with the DUDe.    %Despite being multi-tier, all of the work in UL analysis is limited to SISO HetNets only and no work has considered biasing in the uplink  (except \cite{Geff_DUDE} which is used to achieve DUDe).    
%The authors of \cite{mischa1} studied a heterogeneous network where the downlink association is based on DL received power and the UL association is based on the pathloss. With the DUDe new network boundaries in the UL and DL are formed as shown in Fig. 1., where a user between the DL and UL cell boundaries will be connected to small cell in the UL and to the macro cell in the DL. Here, it is important to mention that this study was performed on real-world network optimization/planning tools. The interference limited UL capacity of this decoupled network increases more than 50\% while the downlink capacity remains the same. The gain in UL capacity is due to two reasons. First, the UE2 has better channel (low pathloss) to small cell , secondly, UE2 causes less interference to the macro cell. To further validate this improvement in UL capacity an analytical study based on stochastic geometry for two-tier network was reported in \cite{katerina}, and \cite{mischa2} where it is shown that as the density of small BSs increases compared to macro BS, more and more UEs convert to decoupled access, which increases the UL coverage probability of the network. In \cite{Geff_DUDE} and \cite{New_DUDe} the analytical study is extended to K-tier heterogeneous networks and the stochastic geometry based results have been derived for coverage and rate.
%\newtheorem{lem}{Lemma}
%\begin{lem}
%Hello
%\end{lem}
\subsection{Contributions and Outcomes}
The main challenge in modeling the UL multi-antennas HetNets, in addition to the generic challenges discussed above, is to select an analytically tractable technique from the number of possible multi-antenna techniques. 
%We study the UL performance of completely random 2-tier network, where BSs have multiple antennas while UEs have single antenna.
We consider maximal ratio combining (MRC) at the BSs and assume that the channel is perfectly known at the receiver. A receiver has knowledge about the channel between the transmitter and itself, but it does not have any knowledge about the interfering channel. Furthermore, we consider power control in the UL, which partially compensates for the path-loss  \cite{Geff_DUDE}, \cite{UL_tier1}. We consider Rayleigh fading in addition to path-loss\footnote{For the sake of simplicity, we do not consider shadowing in this work. Shadowing in similar setup can be found in \cite{Geff_DUDE} and \cite{marco}.}. 

We use a cell association technique with biasing, which can be used in any MIMO HetNets. This association completely decouples the DL and UL association, and is generic and simple. Cell biasing in the UL can be used to balance the load across the tiers. This association scheme is motivated by the technique used in \cite{katerina} for SISO HetNets. Due to the DUDe, users are divided into three disjoint groups as shown in Fig. \ref{system_model}; (I) users attached to the MBS both in the DL and UL, (II) users attached to the FBS both in the DL and UL, and (III) users attached to the MBS in the DL and FBS in the UL. The gain in the UL performance comes from the last kind of users because they have strong connection to the FBS (low path-loss) and they create less interference to the MBS (due to larger distance). 
\begin{figure}
	\centering
		\includegraphics[scale=0.6]{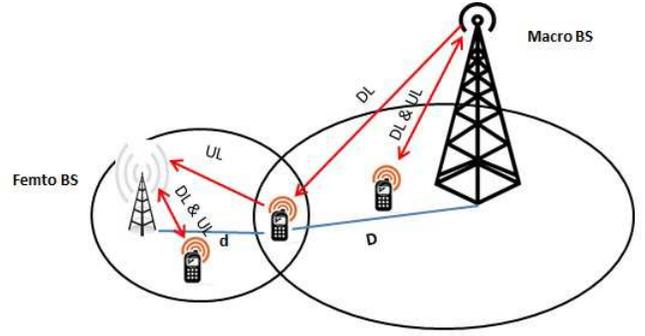}
		\caption{System Model}
	\label{system_model}
	\vspace{-1.7em}
\end{figure}

In this paper, we study both the SIR and rate coverage probability of a two tier network where the association is based on DL and UL decoupling. %We consider multiple antennas MBSs with UL biasing, whereas previous studies on DL and UL decoupling only consider SISO network. 
The novel and insightful findings of this paper are as follows:
\begin{itemize}
\item{} The gain in term of SIR coverage probability provided by the DUDe association over a no-DUDe association (association based on DL maximum received power averaged over fading) decreases as the difference in the number BS's antennas in the femto and macro-tier increases.  When the number of MBS antennas is larger than that of FBS, the association region of a MBS is enlarged due to the larger beamforming gain provided by the MBS. As a result of which UEs closer to the FBSs become associated with MBSs. These boundary UEs, which are connected to macro-tier, create strong interference at nearby FBSs when they transmit to their serving MBSs. On the other hand, when both tiers have the same beamforming gain, the coverage region of both tiers are the same and the interference created by the boundary UEs is not that strong. Thus, the DUDe gain over No-DUDe is high when both tier have the same beamforming.
%In the UL all UEs transmit with the same power and for SISO network the coverage region of a FBS and a MBS is the same, therefore, the gain of DUDe is maximum.
%However, with increasing the number of antennas at MBSs, the beamforming gain of a MBS increases due to which its coverage region increases, which result in reducing the gain of DUDe. 
\item{} It has been shown in \cite{mischa1, katerina, Geff_DUDE, New_DUDe} that DUDe association improves the load balance and provides fairness in the UL performance of different UEs. In \cite{Geff_DUDE} it is shown that in the UL the optimal rate coverage is provided by the minimum path-loss association. However, we observe that in the SIMO network DUDe association does not completely solve the load imbalance problem and the optimal rate coverage is not provided by the minimum path-loss association. In the SIMO network, this load imbalance problem comes from the different beamforming gain of the femto and macro-tier, therefore, we still need biasing towards femto-tier to balance the load. We show that when the beamforming gain of the macro-tier is high as compared to the femto-tier then biasing towards femto-tier improves the rate coverage probability. 
%\item{} Increasing antennas at MBSs without using biasing can degrade the rate coverage of the network. Therefore, in order to obtain the gain of multiple antennas we need a suitable biasing towards femto-tier.
%\item{} Network using DUDe association has the ability to handle interference as compared to the one not using DUDe association. Increasing the the density of macro-tier can improve the rate coverage of the DUDe network, which is not the case for no-decoupling case. 
%It has been found that when the density of the FBS increases,  more and more users become attached to femto-tier. We observe that the gain due to decoupling is maximum when both macro and femto BSs have same number of antennas. When the number of antennas at MBS is large as compared to that of FBS, using a biasing towards femto-tier improves the rate coverage probability, which is contrary to SISO network where the rate coverage is maximum when the association is based on minimum path-loss \cite{Geff_DUDE}. When the association is based on the DL received power, increasing the density of the MBS increases the rate coverage probability of the macro-users while it decreases the rate coverage probability of the femto-users. On the contrary, when the association is based on DUDe then increasing the density of the MBS increases the rate coverage of all users. This happens due to the ability of DUDe to efficiently handle interference. We also quantify the effect of path-loss on the the rate coverage probability.
\end{itemize}
\begin{figure*}[!t]
% ensure that we have normalsize text
\normalsize
%\vspace*{4pt}
%\hrulefill
% Store the current equation number.
%%\setcounter{MYtempeqncnt}{\value{equation}}
% Set the equation number to one less than the one
% desired for the first equation here.
% The value here will have to changed if equations
% are added or removed prior to the place these
% equations are referenced in the main text.
\setcounter{equation}{0}
\begin{equation}
{\bf{Y}}_{K_0}= \sqrt{P_{0}X_{{K}}^{{\alpha_K\left(\eta-1\right)}}}{\bf{h}}_{{K}_0} s_{K_0} + \underbrace{\sum_{i\in\Phi{_{K}'}\backslash u_{0}} \sqrt{P_{0}X_{{K}_i}^{{\alpha}_K\eta}D_{K_i}^{-{\alpha}_K}}{\bf{h}}_{{K}_i}s_{K_i}}_{\text{interference from $K$th tier scheduled UEs}} + \underbrace{\sum_{q\in\Phi{_{J}'}} \sqrt{P_{0}X_{J_q}^{\alpha_J \eta}D_{J_q}^{-\alpha_K}}{{\bf{h}}_{{J}_q}}s_{J_q}}_{\text{interference from $J$th tier scheduled UEs}}  + {\bf{n}}	
\label{sig_ulm1}
\end{equation}
\begin{multline}
{Z}_{K_0}={\bf{h}}_{K_0}^H{\bf{Y}}_{K_0}= \sqrt{P_{0}X_{{K}}^{{\alpha_K\left(\eta-1\right)}}}\left\|{\bf{h}}_{{K}_0}\right\|^2 s_{K_0} +\sum_{i\in\Phi{_{K}'}\backslash u_{0}} \sqrt{P_{0}X_{{K}_i}^{{\alpha}_K\eta}D_{K_i}^{-{\alpha}_K}}{\bf{h}}_{{K}_0}^H{\bf{h}}_{{K}_i}s_{K_i} + \\ \sum_{q\in\Phi{_{J}'}} \sqrt{P_{0}X_{J_q}^{\alpha_J \eta}D_{J_q}^{-\alpha_K}}{\bf{h}}_{{K}_0}^H{{\bf{h}}_{{J}_q}}s_{J_q}  + {\bf{h}}_{{K}_0}^H{\bf{n}}
\label{sig_ulm2}
\end{multline}
\begin{equation}
		\gamma_{K_0}= \frac{P_{0}\left\|{\bf{h}}_{{K}_0}\right\|^2X_{{K}}^{{\alpha_K\left(\eta-1\right)}}}{\sum\limits_{i\in\Phi_{K}'\backslash u_0}P_{0}\left|\frac{{\bf{h}}_{{K}_0}^H {\bf{h}}_{{K}_i}}{\left\|{\bf{h}}_{{K}_0}\right\|}\right|^2 X_{{K}_i}^{{\alpha}_K\eta}D_{K_i}^{-{\alpha}_K}  +\sum\limits_{q\in\Phi_{J}'}P_{0}\left|\frac{{\bf{h}}_{{K}_0}^H {\bf{h}}_{{J}_q}}{\left\|{\bf{h}}_{{K}_0}\right\|}\right|^2X_{J_q}^{\alpha_J \eta}D_{J_q}^{-\alpha_K} + \sigma_n^2}
	\label{sinr}
	\end{equation}
% Restore the current equation number.
%%\setcounter{equation}{\value{MYtempeqncnt}}
% The IEEE uses as a separator
\hrulefill
% The spacer can be tweaked to stop underfull vboxes.
\vspace*{4pt}
\end{figure*}
The rest of the paper is organized as follows, in Section II, we present our system model and assumptions. In Section III, we derive the association probabilities and the distance distribution of a user to its serving BS.  Section IV is the main technical section, where we study the SIR coverage  and the rate coverage of the network. Section V presents  simulations and numerical results, while Section VI concludes the paper and provides further research directions.

The key notations used in this paper are given in Table 1.
\begin{table*}
\caption{List of Notations}
	\centering
		\begin{tabular}{|c|l|}
		\hline
		Notation & Description \\
		\hline 	\hline	
		$\Phi_K$, $\Phi_U$ & PPP of tier $K$ BSs, PPP of UEs  \\	\hline
			$\lambda_K$, $\lambda_U$ & density of tier $K$ BSs, density of UEs  \\ \hline
		$P_K$, $P_U$ & transmit power of each BS of the $K$th tier , transmit power of a UE \\ \hline
		$X_{K}$, $\alpha_K$ & distance between the \emph{typical} UE and the \emph{tagged} BS, path-loss exponent of $K$th tier \\ \hline
		$X_{K_i}$, $X_{J_q}$ & distance between an interfering UE of $K$th and $J$th tier, and their serving BSs respectively \\ 	\hline
		$D_{K_i}$, $D_{J_q}$ & distance between an interfering UE of $K$th and $J$th tier, and the tagged BS respectively\\ \hline
		$\mathcal{A}_K$, $N_K$ & association probability of a \emph{typical} UE to $K$th tier, number of antennas at a $K$th tier BS \\ \hline
		$B$ & bias factor, $B=\frac{B_F}{B_M}$, where $B_F$ and $B_M$ is biasing towards femto-tier and macro-tier respectively\\ \hline
		$\tau_K$, $\rho_K, \eta$ & SIR threshold and rate threshold of $K$th tier, UL power control fraction \\ \hline
		$\mathcal{C}$, $\mathcal{C}_K$ & SIR coverage probability of the network, SIR coverage probability of the $K$th tier\\ \hline
		$\mathcal{R}$, $\mathcal{R}_K$ & rate coverage probability of the network, rate coverage probability of the $K$th tier \\ \hline
		$\Omega_K, \bar{\Omega}_K$, $W$ &  load on a $K$th tier BS, average load on $K$th tier BS, bandwidth in Hz \\ \hline
		$\textbf{h}_{K_0}, \textbf{h}_{K_i}, \textbf{h}_{J_q}$ & complex channel gain between the tagged BS and typical UE, an interfering UE of $K$th and $J$th tier respectively \\ \hline
		\end{tabular}
		%\caption{Table captions go \emph{above} the table}
	\label{table1}
\end{table*}
\section{System Model}
\subsection{Network Model}
We consider a heterogeneous network that consists of macro base stations (MBSs), femto base stations (FBSs) and user equipments (UEs). The location of MBSs, FBSs and UEs are modeled as 2-D independent homogeneous Poisson Point Processes (PPPs). Let $\Phi_M, \Phi_F$, and  $\Phi_U$ represent the PPPs for MBSs, FBSs and UEs respectively. Furthermore, let $\lambda_{M}, \lambda_{F}$, and $\lambda_U$ be the density of $\Phi_M, \Phi_F$, and  $\Phi_U$ respectively.  The transmit power of a MBS and FBS are represented by $P_M$ and  $P_F$ respectively, where $P_M>P_F$. We consider that MBSs have $N_M$ and FBSs have $N_F$ antennas and $N_M \geq N_F$, while UEs have single antenna. Throughout the system model, we only consider inter-cell interference  i.e., a BS schedules a single UE in a given resource block. The analysis is performed for a \textit{typical} user located at the origin and the BS serving this typical user is referred to as the\textit{ tagged} BS \cite{martin_book}.
\subsection{Uplink Power Control}
We consider a fractional power control in the uplink  \cite{power_control}, which partially compensates for path-loss. Let $X_K$ be the distance between a UE and its serving $K$th-tier BS. The UE transmits with $P_U=P_0 X_K^{\eta \alpha_K}$, where $\alpha_K$ is the path-loss exponent of the $K$th-tier, $P_0$ is the transmit power of the UE before applying the UL power control, and $0\leq\eta\leq1$ is the power control fraction. If $\eta=1$, the path-loss is completely inverted by the power control, and  if $\eta=0$, no channel inversion is applied and all UEs transmit with the same power. We do not consider maximum transmit power constraint for tractability of the analysis. However, the analysis can be extended to include the maximum power constraint similar to \cite{hisham} and \cite{marco}.
\subsection{Signal Model}
The received signal vector ${\bf{Y}}_{K_0}$ at a tagged BS when a typical UE $u_0$ is served by a $K$th tier BS having $N_K$ antennas is given by \eqref{sig_ulm1} (at the top of this page), 
where $\alpha_K$ is the path-loss exponent of $K$th tier $(\alpha_K>2)$; ${\bf{h}}_{{K}_i}=\left[h_{K_1}, h_{K_2}, \ldots h_{K_{N_K}}\right]^T$ is the complex channel gain and the magnitude of each $h_i$ follows Rayleigh distribution (we assume Rayleigh fading channel); $X_{J_q}$ represents the Euclidean distance between the $q$th  UE of the $J$th tier and its serving BS; $D_{J_q}$ is the Euclidean distance between the $q$th interfering UE of the $J$th tier to the tagged BS; $s_{J_q}$ is the signal transmitted by the $q$th UE of the $J$th tier having unit power; ${\bf{n}}=\left[n_1,n_2,\cdots,n_{N_K}\right]^T$ is the vector of complex additive white Gaussian noise at the tagged BS; $\Phi{_{K}'}$ and $\Phi{_{J}'}$ represent the point processes formed by the thinned PPP of the scheduled UEs of the $K$th and $J$th tier respectively.   Since, we assume multiple antennas' BS, we apply a receiver combiner $\bf{g}_0$ to $s_{K_0}$ of a typical UE. By using maximal ratio combining (MRC), ${\bf{g}}_0={\bf{h}}_{{M}_0}^H$, \eqref{sig_ulm1} can be written as in \eqref{sig_ulm2} (at the top of this page).
	Similarly, the SINR $\gamma_{K_0}$ at the tagged BS $K_0$ can be written as in \eqref{sinr}, available at the top of this page,
	%\begin{equation}
		%{SINR}^{UL}_{TMBS}= \frac{P_{U}\left\|{\bf{h}}_{{M}_0}\right\|^4X_{{M}_0}^{-\alpha}}{\sum\limits_{i\in\Phi_{U}'\backslash u_0}P_{U}\left|{\bf{h}}_{{M}_0}^H {\bf{h}}_{{M}_i}\right|^2X_{{M}_i}^{-\alpha} + \left\|{\bf{h}}_{{M}_0}\right\|^2\sigma^2}\nonumber %{\sum_ P_{U}h_{MF}^i\left[X_{MF}^i\right]^{-\alpha} + \sigma^2} 	
	%\label{sinrulm},
	%\end{equation}
	%which can be further simplified to
	%\begin{equation}
		%\gamma_{K_0}= \frac{P_{0}\left\|{\bf{h}}_{{K}_0}\right\|^2X_{{K}}^{{\alpha_K\left(\eta-1\right)}}}{\sum\limits_{i\in\Phi_{K}'\backslash u_0}P_{0}\left|\frac{{\bf{h}}_{{K}_0}^H {\bf{h}}_{{K}_i}}{\left\|{\bf{h}}_{{K}_0}\right\|}\right|^2 X_{{K}_i}^{{\alpha}_K\eta}D_{K_i}^{-{\alpha}_K}  +\sum\limits_{q\in\Phi_{J}'}P_{0}\left|\frac{{\bf{h}}_{{K}_0}^H {\bf{h}}_{{J}_q}}{\left\|{\bf{h}}_{{K}_0}\right\|}\right|^2X_{J_q}^{\alpha_J \eta}D_{J_q}^{-\alpha_K} + \sigma_n^2},
	%\label{sinr}
	%\end{equation}
	where $\left\|{\bf{h}}_{{K}_0}\right\|^2 \sim \mathrm{Gamma}\left(N_K,1\right)$, whereas $\left|\frac{{\bf{h}}_{{K}_0}^H {\bf{h}}_{{K}_i}}{\left\|{\bf{h}}_{{K}_0}\right\|}\right|^2$ and $\left|\frac{{\bf{h}}_{{K}_0}^H {\bf{h}}_{{J}_q}}{\left\|{\bf{h}}_{{K}_0}\right\|}\right|^2$ both follow exponential distribution \cite{geff_adhoc}.
%In DL the SINR at a typical user (TU) when the serving BS is a femto BS $(FBS_0)$ becomes
%\begin{equation}
		%{SINR}^{DL}_{TU}= \frac{P_{F}\left|h_{{F}_0}\right|^2X_{{F}_0}^{-\alpha}}{\sum\limits_{i\in\Phi_{F}\backslash FBS_0}P_{F}\left|h_{{F}_i}\right|^2X_{{F}_i}^{-\alpha}+\sum\limits_{j\in\Phi_{M}}P_{M}\left|\frac{{\bf{h}}_{{M}_j}{\bf{h}}_{{M}_j}'^H}{\left\|{\bf{h}}_{{M}_j}'\right\|}\right|^2X_{{M}_i}^{-\alpha} + \sigma^2},
	%\label{sinrdl}
	%\end{equation}
%whereas if the serving BS is a MBS then the SINR becomes
	%\begin{equation}
		%{SINR'}^{DL}_{TU}= \frac{P_{M}\left\|{\bf{h}}_{{M}_0}\right\|^2X_{{M}_0}^{-\alpha}}{\sum\limits_{i\in\Phi_{M}\backslash MBS_0}P_{M}\left|\frac{{\bf{h}}_{{M}_i}{\bf{h}}_{{M}_i}'^H}{\left\|{\bf{h}}_{{M}_i}'\right\|}\right|^2X_{{M}_i}^{-\alpha}+ \sum\limits_{j\in\Phi_{F}}P_{F}\left|h_{{F}_j}\right|^2X_{{F}_j}^{-\alpha} + \sigma^2}.
	%\label{sinrdl_m}
	%\end{equation}
We assume high density for UEs such that each BS has at least one UE in its association region and UEs always have  data to transmit in the UL (saturated queues). Throughout the paper the $K$th tier will always be the serving tier of the typical UE while $J$th tier will be the interfering tier. We will use the terms UE and user, and typical user and random user  interchangeably.
\subsection{Cell Association}
The long term average received power (accounting for beamforming gain) at a typical UE when a $K$th tier BS transmits is $P_K N_K X_K^{-\alpha_K}$.  Similarly, in the UL, the long term average received power at a typical $K$th tier BS is $P_0 N_K X_K^{-\alpha_K}$ (before employing UL power control).
In the DL, a UE is associated to a BS from which it receives the maximum average power, while in the uplink it is associated to a BS that receives the maximum average power. In the UL, each UE has the same transmit power, so the association is actually related to the number of antennas and the path-loss.  Due to the cell association criterion, there are three sets of UEs: 1) UEs connected to the MBSs both in the DL and the UL, 2) UEs associated to the MBSs in the DL and FBSs in the UL, and 3) UEs connected to the FBSs both in the DL and the UL as shown in Fig. \ref{system_model}. In the DL, the load imbalance problem arises due to the high transmit power and beamforming gain of the MBS as compared to the FBS, whereas in the UL it is only due to the larger number of antennas at the MBS. In order to balance the load among the macro-tier and femto-tier in the UL, we use bias factor $B=\frac{B_F}{B_M}$, where $B_F$ and $B_M$ are the bias towards femto- and macro-tier respectively. A biasing $B>1$ offloads UEs from the macro-tier to the femto-tier, $B<1$ offloads UEs form the femto-tier to the macro-tier, and $B=1$ means no biasing. The association criterion is based on long-term average biased-received power and the UEs in different region can be written as:
\begin{itemize}
\item Case1- UEs connected to MBS both in the UL and DL:
\begin{multline} 
 \left\{\underbrace{\left(P_{M}N_M X_M^{-\alpha_M}>P_{F} N_F X_F^{-\alpha_F}\right)}_{\text{DL association rule}}\bigcap \right. \\ \left. \underbrace{\left(N_M B_M X_M^{-\alpha_M}> N_F B_F  X_F^{-\alpha_F}\right)}_{\text{UL association rule}}\right\}, \nonumber
\end{multline}
\item Case2- UEs connected to MBS in the DL and FBS in the UL:
\begin{multline}
 \left\{\left(P_{M} N_M X_M^{-\alpha_M} >   P_{F} N_F X_F^{-\alpha_F} \right) \bigcap \right. \\ \left.  \left(N_M B_M X_M^{-\alpha_M}\leq N_F B_F  X_F^{-\alpha_F}\right)\right\}, \nonumber
\end{multline} 
\item Case3- UEs connected to FBS both in the UL and DL:
\begin{multline}
\left\{\left(P_{M}N_M X_M^{-\alpha_M} \leq P_{F}  N_F X_F^{-\alpha_F}\right)  \bigcap \right. \\ \left. \left(N_M B_M X_M^{-\alpha_M }\leq N_F B_F X_F^{-\alpha_F}\right)\right\}. \nonumber
\end{multline}
\end{itemize}
\section{Preliminaries}
In this section, we find the association probabilities of UEs and the distance distribution of a UE to its serving BS. These will be required in the next section to find the SIR coverage and rate coverage of the network.
\subsection{Association Probability}
In this subsection, we find the association probabilities of the UEs.
\begin{lem}
The probability that a typical UE is associated with the MBS both in the UL and the DL is given by
\begin{equation}
	\mathbb{P}\hspace{-0.1em}\left(case 1\right)\hspace{-0.3em}=\hspace{-0.3em} 2\pi\lambda_M \hspace{-0.3em}\int_0^{\infty}\hspace{-0.7em}X_M e^{-\pi\left[\lambda_F \Upsilon_1^{2/\alpha_F}\hspace{-0.1em}\left(\hspace{-0.1em}X_M^{\alpha_M/\alpha_F}\hspace{-0.1em}\right)^2\hspace{-0.2em}+\lambda_M X_M^2\hspace{-0.1em}\right]}\hspace{-0.1em}\mathrm{d_{X_M}}\hspace{-0.1em},
	\label{prob_case1}
	\end{equation}
where for  $\frac{B_F}{B_M} \geq \frac{P_F}{P_M}$, $\Upsilon_1=\frac{B_F N_F}{B_M N_M}$ and for $\frac{B_F}{B_M} < \frac{P_F}{P_M}$, $\Upsilon_1=\frac{P_F N_F}{P_M N_M}$. The association probability is independent of the density of the UEs.
\label{lem_DLUL_marco}
\end{lem}
\begin{proof}
See Appendix A.
\end{proof}
	%\section{Proof of Lemma 1}
%\end{section}
%\newtheorem{lem}{Lemma}
\begin{lem}
The probability that a typical UE is associated with a MBS in the DL and a FBS in the UL is
\begin{multline}
\mathbb{P}\hspace{-0.1em}\left(case 2\right)\hspace{-0.3em}=\hspace{-0.3em} 2\pi\lambda_F\hspace{-0.3em}\left[\int_0^{\infty}\hspace{-0.7em}X_F e^{-\pi\left[\lambda_M \Upsilon_1'^{2/\alpha_M}\left(X_F^{\alpha_F/\alpha_M}\right)^2+\lambda_F X_F^2\right]}\right.  \\ \left.\mathrm{d_{X_F}} -  \int_0^{\infty}\hspace{-0.7em}X_F e^{-\pi\left[\lambda_M\Upsilon_2'^{2/\alpha_M}\left(X_F^{\alpha_F/\alpha_M}\right)^2+\lambda_F X_F^2\right]}\mathrm{d_{X_F}}\right],
\label{prob_case2}
\end{multline}
where when  $\frac{B_F}{B_M} \geq \frac{P_F}{P_M}$, then $\Upsilon_1'=\frac{B_M N_M}{B_F N_F}$ and $\Upsilon_2'=\frac{P_M N_M}{P_F N_F}$ and when $\frac{B_F}{B_M} < \frac{P_F}{P_M}$ then $\Upsilon_1'=\frac{P_M N_M}{P_F N_F}$ and $\Upsilon_2'=\frac{B_M N_M}{B_F N_F}$.
%\begin{equation}
	%\mathbb{P}\left(case 2\right)= \frac{\lambda_F}{\lambda_F+\left(\frac{N_M}{B}\right)^{2/\alpha}\lambda_M}-\frac{\lambda_F}{\lambda_F+\left(\frac{P_M N_M}{P_F}\right)^{2/\alpha}\lambda_M}  
	%\label{prob_case2}.
	%\end{equation}
	\label{lem_DL_macro_UL_femto}
\end{lem}
\begin{proof}
The proof follows similar steps as Lemma 1. 
\end{proof}
\begin{comment}
However, we provide the proof for completeness. The probability that a UE is connected to a macro BS in DL and femto BS in the UL while using UL biasing is given by
\begin{equation}
	\mathbb{P}\left[\left\{P_{M} N_M X_M^{-\alpha_M} >   P_{F} N_F X_F^{-\alpha_F} \right\} \bigcap  \left\{B_M N_M X_M^{-\alpha_M}\leq B_F N_F X_F^{-\alpha_F}\right\}\right] \nonumber
	\label{caseb2}
	\end{equation}
where	it can be easily noticed that for  $\frac{B_F}{B_M} \geq \frac{P_F}{P_M}$ the common event is $\frac{P_F N_F }{P_M N_M}X_F^{-\alpha_F}<X_M^{-\alpha_M}<\frac{B_F N_F }{B_M N_M} X_F^{-\alpha_F}$, and after some simplification its probability is given by	
\begin{multline}
	\mathbb{P}\left(case 2\right)=\mathbb{P}\left(a \leq X_M \leq b\right) = \int_{0}^{\infty}\left(F_{X_M}\left(b\right)-F_{X_M}\left(a\right)\right)f_{X_F}\left(X_F\right)\mathrm{d_{X_F}}, \nonumber
	\end{multline}
	where $a=\left(\frac{P_M N_M}{P_F N_F}\right)^{1/\alpha_M}X_F^{\alpha_F/\alpha_M}$ and $b=\left(\frac{B_M N_M}{B_F N_F }\right)^{1/\alpha_M}X_F^{\alpha_F/\alpha_M}$. We plug the values of $F_{X_M}\left(.\right) $ and $f_{X_F}\left(.\right)$ from the null probability of 2D PPP and then evaluate the integral. For $\frac{B_F}{B_M} < \frac{P_F}{P_M}$, we use similar technique.
\end{comment}
\begin{lem}
The probability that a typical UE associates with the FBS both in the DL and the UL can be written as
\begin{equation}
\mathbb{P}\hspace{-0.1em}\left(case 3\right)\hspace{-0.3em}=\hspace{-0.3em}2\pi\lambda_F\hspace{-0.3em}\int_0^{\infty}\hspace{-0.7em}X_F e^{-\pi\left[\hspace{-0.1em}\lambda_M\Upsilon_2'^{2/\alpha_M}\left(\hspace{-0.1em}X_F^{\alpha_F/\alpha_M}\hspace{-0.1em}\right)^2\hspace{-0.1em}+\lambda_F X_F^2\hspace{-0.1em}\right]}\hspace{-0.1em}\mathrm{d_{X_F}}\hspace{-0.1em},
\label{prob_case3}
\end{equation}
where when  $\frac{B_F}{B_M} \geq \frac{P_F}{P_M}$, then $\Upsilon_2'=\frac{P_M N_M}{P_F N_F}$ and when \\ $\frac{B_F}{B_M} < \frac{P_F}{P_M}$ then $\Upsilon_2'=\frac{B_M N_M}{B_F N_F}$.

%\begin{equation}
%\mathbb{P}\left(case 3\right)=\frac{\lambda_F}{\lambda_F+\left(\frac{P_M N_M}{P_F }\right)^{2/\alpha}\lambda_M} 
%\label{prob_case3}.
%\end{equation}
\label{lem_DLUL_femto}
\end{lem}
\begin{proof}
It can be easily proved by following the same steps as in Lemma 1.
\end{proof}
%In the same way, the probability that a user connects to  femto-tier both in the DL and the UL is given as
%\begin{equation}
	%Pr\left\{\left\{P_{M}N_M X_M^{-\alpha} \leq P_{F} X_F^{-\alpha}\right\}  \bigcap  \left\{N_M X_M^{-\alpha}\leq B X_F^{-\alpha}\right\}\right\}. \nonumber
	%\end{equation}
	%Here, the common event is $P_{M}N_M X_M^{-\alpha} \leq P_F X_F^{-\alpha}$ and the proof follows similar steps like Lemma 1.
	%\end{proof}
From Lemma 1, 2 and 3, the tier-association probabilities in the UL can be easily obtained.	Thus the probability that a typical UE is associated with $K$th-tier BS is given by
\begin{equation}
\mathcal{A}_K \hspace{-0.1em}=\hspace{-0.1em} 2\pi\lambda_K\hspace{-0.2em} \int_0^{\infty}\hspace{-0.5em}X_K e^{-\pi\left[\lambda_J\Upsilon^{2/\alpha_J}\left(X_K^{\alpha_K/\alpha_J}\right)^2+\lambda_K X_K^2\right]}\mathrm{d_{X_K}}
\label{assoc_prob_gen}
\end{equation}
where $K , J \in  \left\{M,F\right\}$ and $K\neq J$ and for $\frac{B_F}{B_M} \geq \frac{P_F}{P_M}$, $\Upsilon=\frac{B_J N_J}{B_K N_K}$ and for $\frac{B_F}{B_M} < \frac{P_F}{P_M}$, $\Upsilon=\frac{P_J N_J}{P_K N_K}$. It is important to mention that the condition $\frac{B_F}{B_M} < \frac{P_F}{P_M}$ in Lemma \ref{lem_DLUL_marco}, \ref{lem_DL_macro_UL_femto}, \ref{lem_DLUL_femto} and \eqref{assoc_prob_gen} is very unlikely to be true because usually we need to offload the UEs towards femto-tier instead of macro-tier. However, we specifically mentioned it so that the expression in Lemma \ref{lem_DLUL_marco}, \ref{lem_DL_macro_UL_femto}, \ref{lem_DLUL_femto} and  \eqref{assoc_prob_gen} holds for the entire range of the bias $B$.

For $\alpha_K=\alpha_J=\alpha$, \eqref{assoc_prob_gen} simplifies to  
\begin{equation}
\mathcal{A}_K = \frac{\lambda_K}{\lambda_K+\Upsilon^{2/\alpha}\lambda_J}.
\label{assoc_prob_simp}
\end{equation}
The probability that a typical UE associates to the $K$th tier increases with increasing the density of $K$th tier BS, or biasing towards $K$th tier or placing more antennas at $K$th tier BSs. However, the increase due to biasing and beamforming gain is not the dominant factor due to the presence of the exponent $2/\alpha$ where $\alpha >2$.
\subsection{Distance Distribution to the Serving BS}
In this subsection, we find the distance distribution of the scheduled user to the serving BS. 
\begin{lem} 
The distribution of the distance $X_K$ between the typical UE and the  tagged BS is
\begin{multline}
f_{X_K}\left(X_K\right)=\frac{2\pi\lambda_K}{\mathcal{A}_K}X_K \times  \\ \exp\hspace{-0.1em}\left\{\hspace{-0.1em}-\pi \hspace{-0.1em} \left(\hspace{-0.1em}\lambda_K X_K^2\hspace{-0.1em}+\hspace{-0.1em}\lambda_J\hspace{-0.1em}\left(\hspace{-0.1em}\frac{B_J N_J}{B_K N_K}\hspace{-0.1em}\right)^{2/\alpha_J}\hspace{-0.3em}X_K^{2\left(\alpha_K/\alpha_J\right)}\hspace{-0.1em}\right)\hspace{-0.1em}\right\},
\label{dist_distribution}
\end{multline}
where $K,J\in\left\{M,F\right\}$ , $K\neq J$, and $\mathcal{A}_K$ is the tier association probability.
\end{lem}
\begin{proof}
We provide the proof in Appendix B.
\end{proof}
%
%\begin{equation}
%f_{X_k}\left(x\right)=\frac{2\pi \lambda_k x}{A_k}\exp\left\{-\pi\left[\lambda_k+\lambda_j\left(\frac{B_j N_j}{B_k N_k}\right)^{2/\alpha}\right]x^2\right\},
%\label{dist_distribution_simp}
%\end{equation}
%
%The distance distribution depends on the beamforming gain, density and the bias factor of both femto- and macro-tier. Increasing the density of either tier decreases the coverage region of both tiers BS, which reduces the distance between a typical user and its tagged BS. On the other hand, increasing either the biasing towards a tier or the number of BS antennas of a tier enlarges the coverage region of that tier BSs only.% due to which the distance between a typical user and its serving BS increases. %However, this increase in coverage due to biasing and beamforming is only in the power of $2/\alpha$ where $\alpha>2$.  

\begin{remark}
It is important to mention that the distance distribution of an interfering UE to its serving BS is different from the distribution of the typical UE and the tagged BS because the distance between an interfering UE and its serving BS is upper bounded by a function of the distance between an interfering UE and the tagged BS. Specifically, let both the typical UE $u_0$ and an interfering UE $u_i$  belong to the $K$th tier and let the distance between $u_i$ and its serving BS be $X_{K_i}$, and $D_{K_i}$ be the distance between $u_i$ and the tagged BS then $0\leq X_{K_i}\leq D_{K_i}$. Similarly, if $u_i$ belongs to the $J$th tier (interfering tier) and the distance between $u_i$ and its serving BS is $X_{J_i}$ and the distance between $u_i$  and the tagged BS is $D_{J_i}$ then $0\leq X_{J_i}\leq \left(\frac{N_J B_J D_{J_i}^{\alpha_K}}{N_K B_K}\right)^{1/\alpha_J}$. 
\label{remark1}
\end{remark}
\begin{remark}
Based on the association rule in the previous section, we define the interference boundary here. For a UE who is associated to $K$th tier and the association distance is $X_K$, the interference boundary $I_{X_J}$ for the $J$th tier is given by $I_{X_J}=X_J>\left(\frac{N_J B_J}{N_K B_K}\right)^{1/\alpha_J}X_K^{\alpha_K/\alpha_J}$. %between the  and distanceSimilar to Remark \ref{remark1} there is a relationship between $X_K$ (the distance between the typical UE $u_0$ and the tagged BS) and the distance between an interfering UE $u_i$ and the tagged BS. If both $u_0$ and u_i$ belongs to the same $K$th tier then $X_K \leq D_{K_i}\leq\infty$ and if $u_i$ belongs to $J$th tier then $\left(\frac{N_J B_J X_K^{\alpha_K}}{N_K B_K}\right)^{1/\alpha_J} \leq D_{J_i} \leq \infty$.
\label{remark2}
\end{remark}
%\begin{remark}
%Similar to Remark \ref{remark1} there is a relationship between $X_K$ (the distance between the typical UE $u_0$ and the tagged BS) and the distance between an interfering UE $u_i$ and the tagged BS. If both $u_0$ and $u_i$ belongs to the same $K$th tier then $X_K \leq D_{K_i}\leq\infty$ and if $u_i$ belongs to $J$th tier then $\left(\frac{N_J B_J X_K^{\alpha_K}}{N_K B_K}\right)^{1/\alpha_J} \leq D_{J_i} \leq \infty$.
%\label{remark2}
%\end{remark}
Thus, both Remark \ref{remark1} and Remark \ref{remark2} define the regions where the interfering UEs can be located and these regions come due to the association rule defined in the previous section. 
%%%%%%%%%%%%%%%%%%%%%%%%%%%%%%%%%%%
%%%%%%%%%%%%%%%%%%%%%%%%%%%%%%%%%%%%
\section{SIR and Rate Coverage Probability}
\subsection{SIR Coverage Probability}
The UL SIR coverage probability can be defined as the probability that the instantaneous UL SIR at a randomly chosen BS is greater than some predefined threshold. %The analysis is performed on a tagged BS located at the origin, which according to Silvnyak's theorem \cite{martin_book} holds true for a generic BS located at any location in the network. 
The UL SIR coverage probability $\mathcal{C}$ of our system model can be written as
\begin{equation}
\mathcal{C} =\mathcal{C}_F \mathcal{A}_F + \mathcal{C}_M \mathcal{A}_M,    
\end{equation}
where $\mathcal{C}_F$, $\mathcal{C}_M$, $\mathcal{A}_F$, and $\mathcal{A}_M$ are the coverage and association probability of femto- and macro-tier respectively. The \textit{K}th-tier coverage probability $\mathcal{C}_K$ for a target SIR $\tau_K$  can be defined as
\begin{equation}
\mathcal{C}_K \triangleq \mathbb{E}_{X_K}\left[\mathbb{P}\left[\mathrm{SIR}_{X_K}>\tau_K\right]\right].
\end{equation}
%which represents the average fraction of cell area that is in coverage at any time.
%%%%%%%%%%%%%%%%%%%%%%%%%%%%%%%%%%%%%
%%%%%%%%%%%%%%%%%%%%%%%%%%%%%%%%%%%%%
%\subsection{Femto-tier Coverage probability}

In the UL, the interfering UEs do not constitute a homogeneous PPP due to the correlation among the interfering UEs. This correlation is due to the orthogonal channel assignment within a cell and can be better modeled by a soft-core process \cite{soft}. However, soft core processes are generally analytically not tractable \cite{hesham_survey}. Therefore, in most of the UL analysis they approximate it as a single homogeneous PPP (because in the UL the transmit power of the UEs are the same and the association regions of BSs form a Voronoi tessellation) \cite{katerina,mischa2,New_DUDe, hisham}. However, in our system model we can not approximate it as a single homogeneous PPP , due to biasing and different beamforming gain for femto and macro-tier (the association regions of BSs form a weighted Voronoi tessellation). Therefore, we approximate it as  two independent PPPs, i.e., femto-tier interfering UEs constitute one homogeneous PPP while macro-tier interfering UEs constitute another homogeneous PPP. However, we do not approximate the interfering UEs as PPPs in the entire 2-D plane but the regions defined in Remark \ref{remark1} and \ref{remark2}. The constraints of Remark \ref{remark1} and \ref{remark2} are taken into consideration in the rest of the analysis.   % It is important to mention that in \cite{Geff_DUDE} for a SISO network they modeled the interfering UEs as a inhomogeneous PPP, but we do not use inhomogeneous PPP for it because it complicates the analysis (due to multiple antennas at MBSs) and the tractability is lost. 

The channel $\textbf{h}_{K_0}$ follows  $\mathrm{Gamma}\left(N_K,1\right)$, therefore, we need to find the higher order derivative of the Laplace transform of the interference, which is a common problem in MIMO transmission in the PPP network. In the literature, different techniques have been used to simplify the $n$th derivative of the Laplace transform. A Taylor expansion-based approximation is used in \cite{taylor_approx} while \cite{special_fun} uses special  functions  to approximate $n$th derivative of the Laplace transform. However, both of these techniques are applicable to ad-hoc networks only. For cellular network, a recursive-technique is used in \cite{latif}, but their final expression is still complicated, therefore, we use Fa$\grave{\text{a}}$ di Bruno's formula \cite{Faa_de_brunos} to find the $n$th derivative of the Laplace transform of the interference.%$\mathcal{L}_I\left(s\right)$ .
\begin{figure*}[!t]
% ensure that we have normalsize text
\normalsize
% Store the current equation number.
\setcounter{MYtempeqncnt}{\value{equation}}
% Set the equation number to one less than the one
% desired for the first equation here.
% The value here will have to changed if equations
% are added or removed prior to the place these
% equations are referenced in the main text.
\setcounter{equation}{12}
\begin{multline}
 \mathcal{L}_I\left(s\right)=\exp\left(\frac{-2\pi s}{\alpha_K-2} \left[\lambda_K\int_0^\infty X_{K_i}^{2-\alpha_K\left(1-\eta\right)}{}_2 \mathrm{F}{}_1\left[1,1-\frac{2}{\alpha_K},2-\frac{2}{\alpha_K}; -s X_{K_i}^{-\alpha_K\left(1-\eta\right)}\right] \times \right. \right. \\  \\ \left. \left. f_{X_{K_i}}\hspace{-0.4em}\left(X_{K_i}\right)\mathrm{d_{X_{K_i}}} \hspace{-0.1em}+\hspace{-0.1em}
  \lambda_J \zeta^{1-2/\alpha_K}\hspace{-0.5em}\int_{0}^\infty \hspace{-1em} X_{J_q}^{2\alpha_J/\alpha_K-\alpha_J\left(1-\eta\right)}{}_2 \mathrm{F}{}_1\left[1,1-\frac{2}{\alpha_K},2-\frac{2}{\alpha_K}; -s \zeta X_{J_i}^{-\alpha_J\left(1-\eta\right)} \right] {f_{X_{J_q}}\left(X_{J_q}\right)\mathrm{d_{X_{J_q}}}}\right]\right).
\label{lap_thm1}
\end{multline}
% Restore the current equation number.
%\setcounter{equation}{\value{MYtempeqncnt}}
% The IEEE uses as a separator
\hrulefill
% The spacer can be tweaked to stop underfull vboxes.
\vspace*{4pt}
\end{figure*}
\begin{figure*}[!t]
% ensure that we have normalsize text
\normalsize
% Store the current equation number.
\setcounter{MYtempeqncnt}{\value{equation}}
% Set the equation number to one less than the one
% desired for the first equation here.
% The value here will have to changed if equations
% are added or removed prior to the place these
% equations are referenced in the main text.
\setcounter{equation}{14}
\begin{multline}
\mathcal{L}_I\left(s\right)=\exp\left(\frac{-2\pi s}{\alpha_K-2} \left[\lambda_K\int_0^\infty X_{K_i}^2 {}_2\mathrm{F}{}_1\left[1,1-\frac{2}{\alpha_K},2-\frac{2}{\alpha_K}; -s\right]f_{X_{K_i}}\left(X_{K_i}\right)\mathrm{d_{X_{K_i}}} \right. \right. + \\ \left. \left. {\lambda_J \zeta^{1-2/\alpha_K}\int_{0}^\infty X_{J_q}^{2\alpha_J/\alpha_K}{}_2\mathrm{F}{}_1\left[1,1-\frac{2}{\alpha_K},2-\frac{2}{\alpha_K}; {-s\zeta} \right]{f_{X_{J_q}}\left(X_{J_q}\right)\mathrm{d_{X_{J_q}}}}}\right]\right)
\label{cor2},
\end{multline}
%\setcounter{equation}{\value{MYtempeqncnt}}
% The IEEE uses as a separator
\hrulefill
% The spacer can be tweaked to stop underfull vboxes.
\vspace*{4pt}
\end{figure*}

We state the coverage probability of a random user associated to a  $K$th tier BS in the following theorem.
\begin{thm}
The UL coverage probability $\mathcal{C}_K$ of a typical user when the serving BS is a $K$th tier BS and the SIR threshold is $\tau_K$ for the system model in Section II is given by %, assuming that the interfering users form a PPP, the signal channel follows a $\chi^2$ distribution and the interference channel follows an exponential distribution is
%\begin{multline}
%\mathcal{C}_M = \frac{2\pi \lambda_M}{A_M}\int_{{X_M}>0}^\infty \frac{1}{2 \Gamma\left(N_M\right)}s^{N_M-1} \left(-1\right)^{N_M-1} \frac{d^{N_M-1}}{ds^{N_M - 1}}\left\{\exp\left\{-\pi \lambda_I s^{2/\alpha}P_U^{2/\alpha}\int_{\frac{X_M^2}{s^{2/\alpha}P_U^{2/\alpha}}}^{\infty}\left(\frac{1}{u^{\alpha/2}+1}\right)\mathrm{du}\right\}\right\} \\ X_M \exp\left\{-\pi X_M^2\left(\lambda_M+\lambda_F\left(\frac{B_F N_F}{B_M N_M}\right)^{2/\alpha}\right)\right\}\mathrm{d_{X_M}},
%\end{multline}
\setcounter{equation}{11}
\begin{multline}
\mathcal{C}_K\left(\tau_K\right) = \frac{2\pi \lambda_K}{\mathcal{A}_K}\hspace{-0.2em}\int_{0}^\infty \hspace{-1em}X_K\exp\left\{-\pi  \left(\lambda_K X_K^2+\lambda_J\left(\zeta\right)^{2/\alpha_J} \times \right. \right.  \\ \left. \left.  X_K^{2\left(\alpha_K/\alpha_J\right)}\right)\right\} \hspace{-0.5em}\sum_{n=0}^{N_K-1} \frac{s^n \left(-1\right)^n}{n!}  \mathcal{L}_I^n\left(s\right) \mathrm{d_{X_K}},
\label{thm1}
\end{multline}
where $s=\tau_K X_{{K}}^{{\alpha_K\left(1-\eta\right)}}$,  $\zeta =  \frac{N_J B_J}{N_K B_K}$, $ \mathcal{L}_I\left(s\right)$ is the Laplace transform of the interference given in \eqref{lap_thm1}, available at the top of this page.  %$\xi = \frac{B_J N_J}{B_K N_K}$ ,
%\begin{multline}
%\mathcal{L}_I\left(s\right) = \exp\left\{ \frac{-2\pi s^{2/\alpha} \tau}{\alpha-2}\left[ \left(\lambda_M-1\right) {}_2 F{}_1\left(1,\frac{\alpha-2}{\alpha},2-\frac{2}{\alpha}; -\tau\right)+ \lambda_F \xi^{2/\alpha-1} \right. \right. \\ \left. \left.   {}_2 F{}_1\left(1,\frac{\alpha-2}{\alpha},2-\frac{2}{\alpha}; -\frac{\tau}{\xi}\right)\right]\right\},\nonumber
%\end{multline}
$\mathcal{L}_I^n\left(s\right)$ represents the $n$th derivative of the $\mathcal{L}_I\left(s\right)$ and to find it we utilize Fa$\grave{\text{a}}$ di Bruno's formula \cite{Faa_de_brunos} 
\begin{multline}
\mathcal{L}_I^n\left(s\right)= \sum\frac{n!}{b_1!b_2!\cdots b_n!} \mathcal{L}_I^k\left(s\right) \left(\frac{f'\left(s\right)}{1!}\right)^{b_1} \left(\frac{f''\left(s\right)}{2!}\right)^{b_2} \\ \cdots\left(\frac{f^n\left(s\right)}{n!}\right)^{b_n},  \nonumber
\end{multline}
where $f\left(s\right)$ is the term inside the exponential of \eqref{lap_thm1} and the summation is to be performed over all different solutions in non-negative integers $b_1, \cdots ,b_n$ of $b_1+2b_2+\cdots+nb_n=n$ and $k=b_1+\cdots+b_n.$ %= \frac{-2\pi s^{2/\alpha} \tau}{\alpha-2}\left[ \left(\lambda_M-1\right) {}_2 F{}_1\left(1,\frac{\alpha-2}{\alpha},2-\frac{2}{\alpha}; -\tau\right)+ \lambda_F \xi^{2/\alpha-1}  {}_2 F{}_1\left(1,\frac{\alpha-2}{\alpha},2-\frac{2}{\alpha}; -\frac{\tau}{\xi}\right)\right]$ and the summation is to be performed over all different solutions in non-negative integers $b_1, \cdots ,b_n$ of $b_1+2b_2+\cdots+nb_n=n$ and $k=b_1+\cdots+b_n.$
\end{thm}
\begin{proof}
See Appendix C.
\end{proof}
We see that as the number of antennas $N_K$ increases, the summation term becomes larger, and after taking the $n$th derivative, the expression becomes very lengthy. Hence, numerically computing the coverage probability is computationally very expensive. % However, the coverage probability of macro-tier only requires the evaluation of a single integral and can be easily computed. 
\subsection{Special Cases}
The SIR coverage in Theorem 1 can be simplified for the following plausible special cases.
\begin{cor}
The $K$th tier SIR coverage probability without UL power control $\left(\eta=0\right)$ is given by \eqref{thm1} while the $\mathcal{L}_I\left(s\right)$ simplifies to 
%\begin{multline}
 %\mathcal{L}_I\left(s\right)=\exp\left(\frac{-2\pi s}{\alpha_K-2} \left[{\lambda_K X_{K}^{2-\alpha_K} {}_2\mathrm{F}{}_1\left[1,1-\frac{2}{\alpha_K},2-\frac{2}{\alpha_K}; -s X_{K}^{-\alpha_K}\right]} \right. \right. + \\ \left. \left. {\lambda_J \zeta^{\frac{2-\alpha_K}{\alpha_J}} X_K^{\frac{2\alpha_K-\alpha_K^2}{\alpha_J}} {}_2\mathrm{F}{}_1\left[1,1-\frac{2}{\alpha_K},2-\frac{2}{\alpha_K}; \frac{-s X_{K}^{-\alpha_K^2/\alpha_J}}{\zeta^{\alpha_K / \alpha_J} }\right]}\right]\right),
%\label{cor1}
%\end{multline}
\setcounter{equation}{13}
\begin{multline}
 \mathcal{L}_I\hspace{-0.2em}\left(\hspace{-0.2em}s\hspace{-0.2em}\right)\hspace{-0.2em}=\hspace{-0.2em}\exp\hspace{-0.2em}\left(\hspace{-0.2em}\frac{-2\pi \tau_K }{\alpha_K-2} \hspace{-0.2em}\left[\hspace{-0.2em}{\lambda_K s^{2/\alpha_K} {}_2\mathrm{F}{}_1\hspace{-0.3em}\left[\hspace{-0.2em}1\hspace{-0.1em},\hspace{-0.2em}1\hspace{-0.2em}-\hspace{-0.2em}\frac{2}{\alpha_K}\hspace{-0.1em},\hspace{-0.2em}2\hspace{-0.2em}-\hspace{-0.2em}\frac{2}{\alpha_K}\hspace{-0.1em};\hspace{-0.2em} - \tau_K \hspace{-0.2em}\right]} \right. \right. \hspace{-0.5em}+ \\ 
\left. \left. {\lambda_J \zeta^{\frac{2-\alpha_K}{\alpha_J}}\hspace{-0.2em} s^{\frac{2+\alpha_J-\alpha_K}{\alpha_J}}\hspace{-0.2em} {}_2\mathrm{F}{}_1\hspace{-0.3em}\left[\hspace{-0.2em}1\hspace{-0.1em},\hspace{-0.2em}1\hspace{-0.2em}-\hspace{-0.2em}\frac{2}{\alpha_K},\hspace{-0.2em}2\hspace{-0.2em}-\hspace{-0.2em}\frac{2}{\alpha_K}; \hspace{-0.2em}\frac{-\tau_K s^{1-\alpha_K/\alpha_J}}{\zeta^{\alpha_K / \alpha_J} }\hspace{-0.2em}\right]}\hspace{-0.2em}\right]\hspace{-0.2em}\right)\hspace{-0.2em},
\label{cor1}
\end{multline}
where $s= X_{{K}}^{\alpha_K}$ and the rest of the variables have the usual meaning. 
\end{cor}
The coverage probability can be found by evaluating just a single integral.
%\begin{proof}
%See Appendix X.
%\end{proof}
%
%
\begin{cor}
The $\mathcal{C}_K$ with full channel inversion $\left(\eta=1\right)$ is given by \eqref{thm1} while the $\mathcal{L}_I\left(s\right)$ simplifies to \eqref{cor2}, available at the top of this page,
%\begin{multline}
%\mathcal{L}_I\left(s\right)=\exp\left(\frac{-2\pi s}{\alpha_K-2} \left[\lambda_K\int_0^\infty X_{K_i}^2 {}_2\mathrm{F}{}_1\left[1,1-\frac{2}{\alpha_K},2-\frac{2}{\alpha_K}; -s\right]f_{X_{K_i}}\left(X_{K_i}\right)\mathrm{d_{X_{K_i}}} \right. \right. + \\ \left. \left. {\lambda_J \zeta^{1-2/\alpha_K}\int_{0}^\infty X_{J_q}^{2\alpha_J/\alpha_K}{}_2\mathrm{F}{}_1\left[1,1-\frac{2}{\alpha_K},2-\frac{2}{\alpha_K}; {-s\zeta} \right]{f_{X_{J_q}}\left(X_{J_q}\right)\mathrm{d_{X_{J_q}}}}}\right]\right)
%\label{cor2},
%\end{multline}
where $s=\tau_K$ while the rest of the parameters remain the same.
\end{cor}
\begin{cor}
For $B_K N_K = B_J N_J$ and $\alpha_K=\alpha_J=\alpha$ the $\mathcal{C}_K$ is given by
\setcounter{equation}{15}
\begin{equation}
\mathcal{C}_K\hspace{-0.2em}\left(\hspace{-0.2em}\tau_K\hspace{-0.2em}\right) \hspace{-0.2em}= \hspace{-0.2em}\frac{2\pi \lambda_K}{\mathcal{A}_K}\hspace{-0.4em}\int_{0}^\infty \hspace{-0.6em}{X_K\hspace{-0.2em}\exp\left\{\hspace{-0.2em}-\pi \lambda X_K^2 \hspace{-0.2em}\right\}\hspace{-0.3em} \sum_{n=0}^{N_K-1}\hspace{-0.3em} \frac{s^n \hspace{-0.1em}\left(\hspace{-0.1em}-1\hspace{-0.1em}\right)^n}{n!}  \mathcal{L}_I^n\hspace{-0.2em}\left(\hspace{-0.2em}s\hspace{-0.2em}\right)\hspace{-0.2em} \mathrm{d_{X_K}}}\hspace{-0.1em},
\label{cor3}
\end{equation}
where $\lambda=\lambda_K+\lambda_J$ and $\mathcal{L}_I\left(s\right)$ is
\begin{multline}
\mathcal{L}_I\left(s\right)=\exp\left(\frac{-2\pi s \lambda} {\alpha-2} \int_0^\infty X_i^{2-\alpha\left(1-\eta\right)} \right. \times \\ \left. {}_2\mathrm{F}{}_1\left[1,1-\frac{2}{\alpha},2-\frac{2}{\alpha}; -sX_i^{-\alpha\left(1-\eta\right)}\right]f_{X_i}\left(X_i\right)\mathrm{d_{X_i}}\right).
\end{multline}
\end{cor}
The coverage probability behaves as if the interference is from a single tier network with density $\lambda=\lambda_K+\lambda_J$.
\begin{cor}
For $N_K=N_J$, $B_K=B_J$, $\alpha_K=\alpha_J=\alpha$, $\tau_K=\tau_J=\tau$ and $\lambda_K=\lambda_J=\lambda$ then the coverage probability is given by
\begin{multline}
\mathcal{C}=\mathcal{C}_K=\mathcal{C}_J=\frac{2\pi \lambda}{\mathcal{A}}\int_{0}^\infty X_K\exp\left\{-2 \pi \lambda X_K^2 \right\} \times \\ \sum_{n=0}^{N_K-1} \frac{s^n \left(-1\right)^n}{n!}  \mathcal{L}_I^n\left(s\right) \mathrm{d_{X_K}},
\label{cor4}
\end{multline}
where $\mathcal{A} = \mathcal{A}_K=\mathcal{A}_J$ and $\mathcal{L}_I\left(s\right)$ is 
\begin{multline}
\mathcal{L}_I\left(s\right)=\exp\left(\frac{-4\pi s \lambda} {\alpha-2} \int_0^\infty X_i^{2-\alpha\left(1-\eta\right)}  \times \right. \\ \left. {}_2\mathrm{F}{}_1\left[1,1-\frac{2}{\alpha},2-\frac{2}{\alpha}; -sX_i^{-\alpha\left(1-\eta\right)}\right]f_{X_i}\left(X_i\right)\mathrm{d_{X_i}}\right).
\end{multline}
\end{cor}
The network coverage probability $\mathcal{C}$ becomes equal to the tier coverage probability $\mathcal{C}_K$, $\mathcal{C}_J$.
\begin{cor}
For $\eta=0$, $B_K N_K = B_J N_J$, $\alpha_K=\alpha_J=\alpha$ the $\mathcal{C}_K$ is given by \eqref{cor3} while the $\mathcal{L}_I\left(s\right)$ simplifies to
\begin{equation}
\mathcal{L}_I\left(s\right)\hspace{-0.2em}=\hspace{-0.2em}\exp\hspace{-0.2em}\left(\hspace{-0.2em}\frac{-2\pi \tau_K s^{2/\alpha} \lambda}{\alpha-2}  {}_2\mathrm{F}{}_1\left[1,1-\frac{2}{\alpha},2-\frac{2}{\alpha}; -\tau_K \right]\hspace{-0.2em}\right)\hspace{-0.2em},
\label{cor5}
\end{equation} 
where  $s= X_K^{\alpha}$ and  $\lambda=\lambda_K+\lambda_J$.
\end{cor}
The coverage probability is in the form of single integral and the interference behaves as if it originates from a single tier network.
\begin{cor}
For $\eta=0$, $N_K = 1$, $\alpha_K=\alpha_J=\alpha$ the $\mathcal{C}_K$ is 
\begin{equation}
\mathcal{C}_K\left(\tau_K\right) \hspace{-0.2em}= \hspace{-0.2em} \frac{\lambda_k}{\mathcal{A}_K \hspace{-0.2em} \left[\lambda_K+\lambda_J\zeta^{-2/\alpha}+\frac{2 \tau_K}{\alpha-2}\mathrm{G}\left(\alpha,\tau_K,\zeta,\lambda_K, \lambda_J\right)\right]}\hspace{-0.1em},
\label{cor6}
\end{equation}
where $\mathrm{G}\hspace{-0.1em}\left(\alpha,\tau_K,\zeta,\lambda_K, \lambda_J\right)\hspace{-0.1em}=\hspace{-0.1em}\lambda_K {}_2\mathrm{F}{}_1\hspace{-0.2em}\left[1,\hspace{-0.1em}1\hspace{-0.1em}-\hspace{-0.1em}\frac{2}{\alpha},\hspace{-0.1em}2\hspace{-0.1em}-\hspace{-0.1em}\frac{2}{\alpha};\hspace{-0.1em} -\tau_K \hspace{-0.2em}\right] \hspace{-0.2em} +\\ \lambda_J \zeta^{2/\alpha-1}{}_2\mathrm{F}{}_1\left[1,1-\frac{2}{\alpha},2-\frac{2}{\alpha};- \frac{\tau_K}{\zeta}\right]$, and $\zeta =  \frac{B_K}{N_J B_J}$.
\end{cor}
The coverage probability reduces to closed form.
\begin{cor}
For $\eta=0$, $N_K = N_J = 1$ $B_K=B_J=1$, $\alpha_K=\alpha_J=\alpha$ the the $\mathcal{C}_K$ can further be simplified to
\begin{equation}
\mathcal{C}_K\left(\tau_K\right) = \frac{1}{1+\frac{2\tau_K}{\alpha-2}{}_2\mathrm{F}{}_1\left[1,1-\frac{2}{\alpha},2-\frac{2}{\alpha}; -\tau_K \right]}
\label{cor7}.
\end{equation}
\end{cor}
The coverage probability becomes density invariant. 
\subsection{Rate Coverage Probability}
In this subsection, we find the rate coverage probability of the network, which is the probability that a randomly chosen user can achieve a target rate or the average fraction of users that achieve the target rate. The rate coverage probability of the network can be written as
\begin{equation}
\mathcal{R} = \mathcal{A}_F\mathcal{R}_F + \mathcal{A}_M\mathcal{R}_M,
\end{equation}
where $\mathcal{R}_F$ and  $\mathcal{R}_M$ are the rate coverage probability, and $ \mathcal{A}_F$ and $ \mathcal{A}_M$ are the association probability of the femto- and macro-tier respectively. The rate coverage $\mathcal{R}_K$ of the $K$th tier when the rate threshold is $\rho_K$ can be written as 
\begin{equation}
\mathcal{R}_K  \triangleq \mathbb{P}\left[\frac{W}{\Omega_K}\log_2\left(1+\mathrm{SIR}_K\right)>\rho_K\right],
\label{rate_cov_basic}
\end{equation}
where $W$ is the frequency resources and $\Omega_K$ is the load on a $K$th-tier BS.  The rate distribution captures the effect of both $\mathrm{SIR}_K$ and load $\Omega_K$, which in turn depends on the corresponding association area. The distribution of the association area is complex and not known. However, by using the association area approximation in \cite{rate_distr}, the probability mass function of the load is given by
\begin{multline}
\mathbb{P}\left(\Omega_K=n\right) = \frac{3.5^{3.5}}{\left(n-1\right)!}\frac{\Gamma\left(n+3.5\right)}{\Gamma\left(3.5\right)}\left(\frac{\lambda_U \mathcal{A}_K}{\lambda_K}\right)^{n-1} \times \\ \left(3.5+\frac{\lambda_U \mathcal{A}_K}{\lambda_K}\right)^{-\left(n+3.5\right)}, n\geq 1,
\label{pmf_load}
\end{multline}
where $\Gamma\left(t\right)=\int_0^\infty x^{t-1}\exp\left(-x\right)\mathrm{dx}$ is a gamma function. %By considering $\mathrm{SIR}$ and $\mathrm{load}$ independent the rate coverage probability $\mathcal{R}_K$ of the $K$th tier is given in the following Theorem.

We state the rate coverage probability $\mathcal{R}_K$ in the following Theorem.
%\begin{thm}
%The femto-tier rate coverage probability $\mathcal{R}_F$ for the system model under consideration and assumptions is given by
%\begin{equation}
%\mathcal{R}_F \left(\rho\right)= \sum_{n\geq1}\frac{3.5^{3.5}}{\left(n-1\right)!}\frac{\Gamma\left(n+3.5\right)}{\Gamma\left(3.5\right)}\left(\frac{\lambda_U A_F}{\lambda_F}\right)^{n-1}\left(3.5+\frac{\lambda_U A_F}{\lambda_F}\right)^{-\left(n+3.5\right)}\mathcal{C}_F\left(2^{\rho n/W}-1\right)
%\label{femto_rate_cov},
%\end{equation}
%where $\mathrm{C}_F$ is given by \eqref{thm_1}.
%\end{thm}
%\begin{proof}
%See Appendix F for the proof.
%\end{proof}
%Similar to the femto-tier coverage probability, we state the macro-tier coverage probability in the following Theorem.
\begin{thm}
The $\mathcal{R}_K$ when the rate threshold is $\rho_K$ for the system model under consideration is given by
\begin{multline}
\mathcal{R}_K \left(\rho_K\right)= \sum_{n\geq1}\frac{3.5^{3.5}}{\left(n-1\right)!}\frac{\Gamma\left(n+3.5\right)}{\Gamma\left(3.5\right)}\left(\frac{\lambda_U \mathcal{A}_K}{\lambda_K}\right)^{n-1}\times \\ \left(3.5+\frac{\lambda_U \mathcal{A}_K}{\lambda_K}\right)^{-\left(n+3.5\right)}\mathcal{C}_K\left(2^{\rho_K n/W}-1\right)
\label{macro_rate_cov},
\end{multline}
where $\mathcal{C}_K$ is given by \eqref{thm1}.
\end{thm}
\begin{proof}
The rate coverage probability of the $K$th tier for threshold $\rho_K$ can be written as
\begin{align}
\mathcal{R}_K\left(\rho_K\right) &= \mathbb{P}\left[\frac{W}{\Omega_K}\log_2\left(1+\mathrm{SIR}_K\right)>\rho_K\right] \nonumber \\ 
&= \mathbb{P}\left[\mathrm{SIR}_K > 2^{\rho_K \Omega_K /W}-1 \right].
\end{align}
By the definition of the SIR coverage probability the above expression becomes
\begin{align}
\mathcal{R}_K\left(\rho_K\right)&=\mathbb{E}_{\Omega_K}\left[\mathcal{C}_K\left(2^{\rho_K \Omega_K /W}-1\right)\right] \nonumber \\ 
&= \sum_{n \geq 1}\mathbb{P}\left(\Omega_K = n \right) \mathcal{C}_K\left(2^{\rho_K n/W}-1 \right).
\end{align}
By putting \eqref{pmf_load} in the above expression, we obtain \eqref{macro_rate_cov}.
\end{proof}
The rate coverage probability expression in \eqref{macro_rate_cov} can be further simplified by using the mean load approximation used in \cite{rate_distr}. The mean load is given by
\begin{equation}
\bar{\Omega}_K = \mathbb{E}\left[\Omega_K\right]=1+\frac{1.28\lambda_U \mathcal{A}_K}{\lambda_K}
\label{average_load},
\end{equation}
where $K \in  \left\{M,F\right\}$. By using the mean load $\bar{\Omega}_K$ the summation over $n$ is removed from \eqref{macro_rate_cov}. 
\section{Results and Discussion}
\begin{figure} 
	\centering
		\includegraphics[scale=0.54]{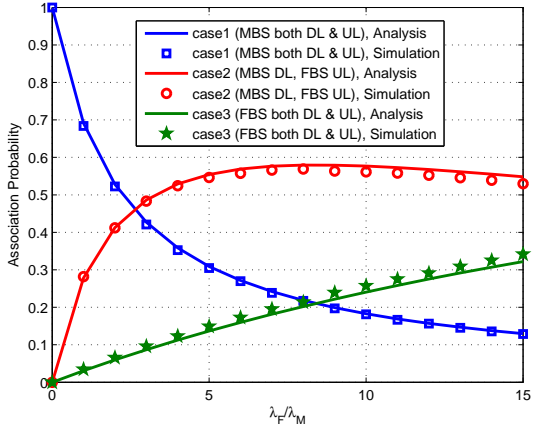}
		\caption{UL Association probabilities vs. $\lambda_F/\lambda_M$, $\left(\alpha=4, N_M=5, N_F=1, B=1\right)$. }
	\label{associatioin1}
	\end{figure}
%%%%%%%%%%%%%%%%%%%%%%%%%%%%%%%%%%%
%\lipsum
%\begin{figure}
    %\centering
    %\begin{minipage}{0.45\textwidth}
        %\centering
        %\includegraphics[width=0.9\textwidth]{associatioin_simulation_analytical.eps} % first figure itself
        %\caption{first figure}
    %\end{minipage}\hfill
    %\begin{minipage}{0.45\textwidth}
        %\centering
        %\includegraphics[width=0.9\textwidth]{AssociationMBS5_Picobias5_final_A.eps} % second figure itself
        %\caption{second figure}
    %\end{minipage}
%\end{figure}
%\lipsum[3]
%%%%%%%%%%%%%%%%%%%%%%%%%%%%%%%%%%%%%%
\begin{figure} 
	\centering
		\includegraphics[scale=0.54]{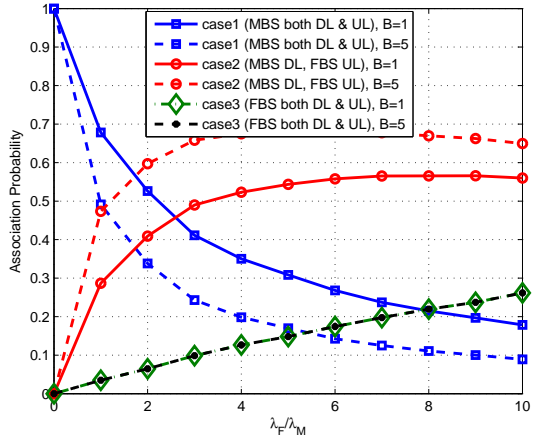}
		\caption{Effect of biasing on the UL association probabilities, $\left(\alpha=4, N_M=5, N_F=1, B=5\right)$. }
	\label{biasing}
	\end{figure}
%\subsection{Validation Through Simulations}
First, we discuss the accuracy of our analysis and system model. MBSs, FBSs and UEs are deployed according to the system model, and we fix $P_M=43$ dBm, $P_F=20$ dBm, $P_0=-100$ dBm/Hz, and $W=10$ MHz. All the densities $\lambda_M, \lambda_F$ and $\lambda_U$ are per square kilometers $/\text{Km}^2$. We consider the same $\mathrm{SIR}$ thresholds  $\left(\tau=\tau_M=\tau_F\right)$, rate thresholds $\left(\rho=\rho_M=\rho_F\right)$ and path-loss exponents $\left(\alpha=\alpha_M=\alpha_F\right)$ for both tiers. 

%The variance of the noise $\sigma_n^2=-110$dBm, and the densities $\lambda_M, \lambda_F$ and $\lambda_U$ are per square kilometers $(/km^2)$.
Fig. \ref{associatioin1} shows the association probabilities of UEs to different cases (mentioned in Section II) versus ratio of $\lambda_F$ and $\lambda_M$, $\left(\lambda_F/\lambda_M\right)$, for the given parameters. The solid lines show analytical results, derived using \eqref{prob_case1}, \eqref{prob_case2}, and \eqref{prob_case3} while marked points are obtained using Monte Carlo simulations. It can be noticed that as the density of the FBS, $\lambda_F$, increases, the number of UEs in $case2$ and $case3$ also increases, whereas the number of UEs in $case1$ decreases. It can  further be noticed that initially the association probability of $case2$ increases very rapidly and reaches a maximum value, $\left(\lambda_F/\lambda_M=7\right)$ , and then starts decreasing because a larger number of UEs become attached to FBSs both in the DL and UL. The figure provides an estimate of the load in different tiers for design engineers. We can observe that at $\lambda_F/\lambda_M=5$, 30$\%$ of the UEs is attached to macro-tier ($case1$) while 70$\%$ of UEs is attached to femto-tier ($case2 + case3$), but if we increase $N_M=25$ and keep the rest of the parameters the same then 50$\%$ of the UEs will be attached to macro-tier and 50$\%$ to femto-tier (using \eqref{assoc_prob_gen}). This shows that even using DUDe and higher density for the femto-tier, we still need to balance the load between the tiers. Therefore, we use biasing to balance the load and the next figure shows the effect of biasing on different UEs' type.

\begin{figure*}
		\begin{subfigure}{0.45\textwidth}
			\centering
			\includegraphics[width=1\linewidth]{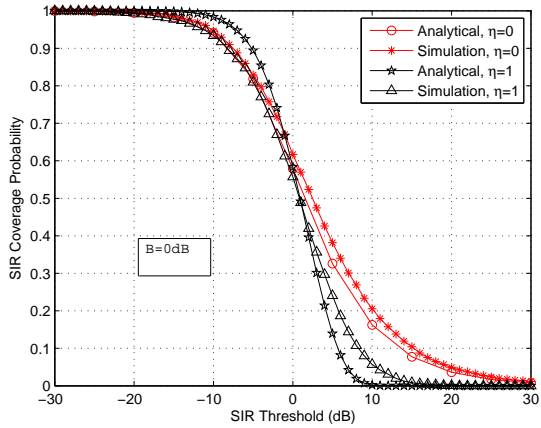}
			\caption{$\lambda_M=3, \lambda_F=10, N_M=4, N_F=2, \alpha=3$}
			\label{sir_ulpc}
		\end{subfigure}\hfill
	\begin{subfigure}{0.45\textwidth}
	  \centering
		 \includegraphics[width=1\linewidth]{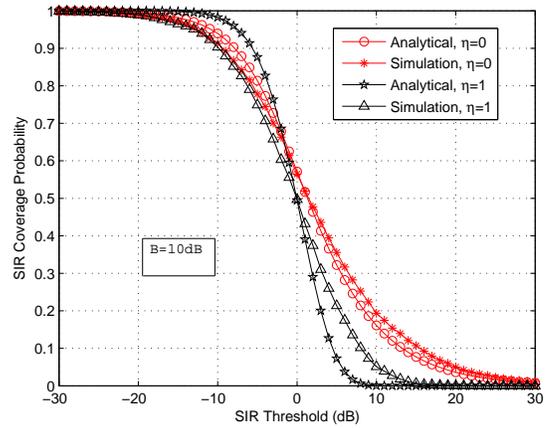}
		 \caption{$\lambda_M=3, \lambda_F=10, N_M=4, N_F=2, \alpha=3$}
		 \label{sir_noulpc}
	\end{subfigure}\hfill
	\begin{subfigure}{0.45\textwidth}
			\centering
			\includegraphics[width=1\linewidth]{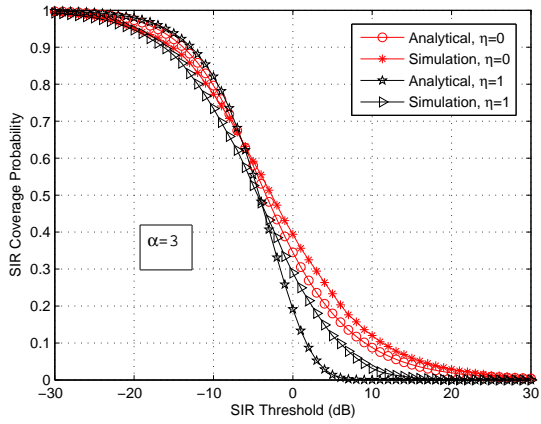}
			\caption{$\lambda_M=1, \lambda_F=4, N_M=N_F=1, B=10\text{dB}$}
			\label{sir_ulpc}
		\end{subfigure} \hfill
	\begin{subfigure}{0.45\textwidth}
	  \centering
		 \includegraphics[width=1\linewidth]{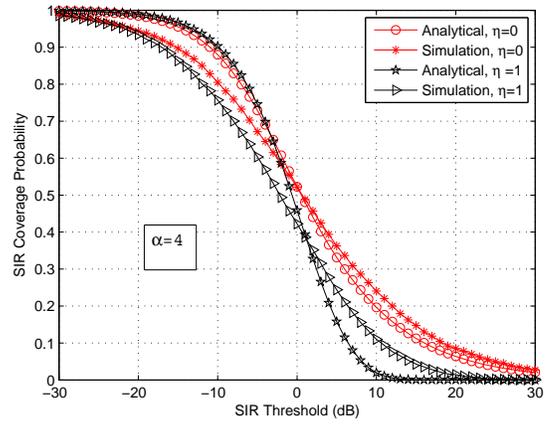}
		 \caption{$\lambda_M=1, \lambda_F=4, N_M=N_F=1, B=10\text{dB}$}
		 \label{sir_noulpc}
	\end{subfigure}\hfill
	\caption{SIR coverage probability simulations vs analytical}
\label{coverage_prob}
\vspace{-1.7em}
\end{figure*}

Fig. \ref{biasing} depicts the effect of biasing on association probabilities. It can  easily be noticed that by using $B=5$ the association probability of $case2$ increases while the association probability of $case1$ decreases. When $B>1$ it offloads the boundary UEs of the macro-tier and these UEs become attached to femto-tier. Similarly, when $B<1$ the boundary UEs of the femto-tier are offloaded to the macro-tier, whereas $B=1$ means no biasing. By changing $B$ we can balance the load among two tiers for optimal performance. 

Fig. \ref{coverage_prob} compares the SIR coverage probability obtained through simulations and analysis for various network parameters.
%Further, Fig. \ref{sir_ulpc} and Fig. \ref{sir_noulpc} show the serving tier SIR coverage probability for power control $\eta=0.5$ and no power control $\eta=0$, respectively. 
%The numerical results in Fig. \ref{sir_ulpc} and Fig. \ref{sir_noulpc} are obtained through \emph{Theorem 1} and \emph{Corollary 1}, respectively. 
It can be noticed that the analysis and simulations curves are close to each other, which shows that the independent homogeneous PPPs approximation of the interfering UEs is reasonably accurate. The gap between the simulation and the numerical curve is due to the homogeneous PPP approximation of the interfering UEs. There is some correlation among the interfering UEs as discussed in Section IV. However, it is quite challenging to model this correlation. Therefore, in most of the UL analysis this correlation is ignored \cite{katerina}, \cite{UL_tier1}, \cite{hisham} and \cite{marco}. In \cite{Geff_DUDE} and \cite{non_homo_letter} the interfering UEs are approximated as non-homogeneous PPP in a SISO network model. However, due to multi-antenna BSs in our system model, we need to find the higher order derivative of the Laplace transform of the interference, and approximating the interfering UEs as non-homogeneous PPP makes the analysis even more involved.

%\begin{comment}
%\begin{figure}
		%\begin{subfigure}{0.5\textwidth}
			%\centering
			%\includegraphics[width=1\linewidth]{simulation_analytical_Femto_tier_alpha3_MBias5_MBS5_A.eps}
			%\caption{Femto-Tier}
			%\label{femto_coverage}
		%\end{subfigure}
	%\begin{subfigure}{0.5\textwidth}
	  %\centering
		 %\includegraphics[width=1\linewidth]{simulation_analytical_Macro_tier_alpha3_MBias5_MBS5_B.eps}
		 %\caption{Macro-Tier}
		 %\label{macro_coverage}
	%\end{subfigure}
	%\caption{SIR coverage probability simulations vs analytical, $\left(\lambda_M=2, \lambda_F=12, \alpha=3, N_M=5\right)$.}
%\label{coverage_prob}
%\end{figure}
%\begin{figure}
	%\begin{subfigure}{0.5\textwidth}
		%\centering
			%\includegraphics[width=1\linewidth]{analyt_sim_mbs6_bw10mhz_user3000_alpha_3.5averagen_lm3_lf12_BF10.eps}
			%\caption{$\lambda_F=4\lambda_M, \lambda_U=3000$}
			%\label{femto_coverage}
		%\end{subfigure}
	%\begin{subfigure}{.5\textwidth}
	  %\centering
		%\includegraphics[width=1\linewidth]{analyt_vs_Sim_mbs6_bw10mhz_user5000_alpha_3.5averagen_lm3_lf24_BF10.eps}
		 %\caption{$\lambda_F=8\lambda_M, \lambda_U=5000$}
		 %\label{macro_coverage}
	%\end{subfigure}
	%\caption{Rate coverage probability simulations vs analytical, $\left(\lambda_M=3, \alpha=3.5, N_M=6, B=10, W=10\text{MHz}\right)$.}
%\label{rate_coverage_comp}
%\end{figure}
%\end{comment}
\begin{figure*}
	\begin{subfigure}{0.45\textwidth}
		\centering
			\includegraphics[width=1\linewidth]{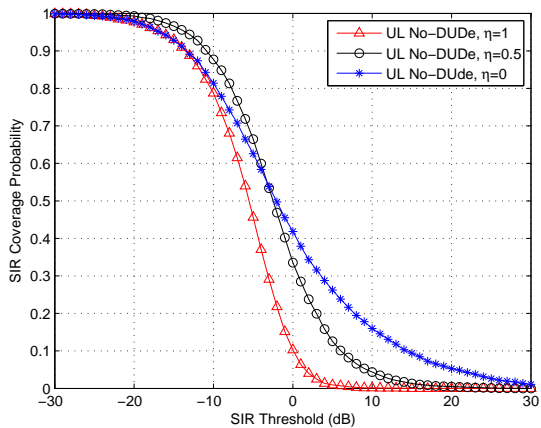}
			\caption{No-DUDe}
			\label{PC_NO_DUDe}
		\end{subfigure}\hfill
	\begin{subfigure}{.45\textwidth}
	  \centering
		\includegraphics[width=1\linewidth]{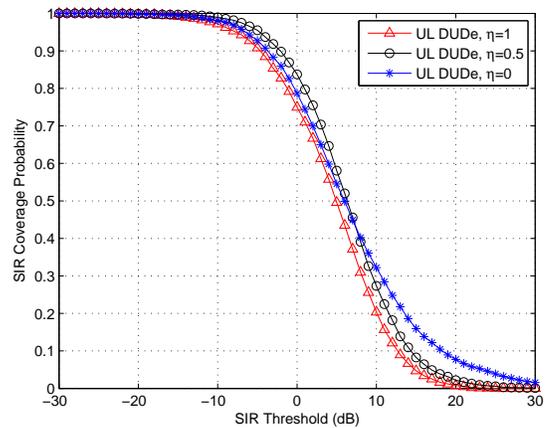}
		 \caption{DUDe}
		 \label{PC_DUDe}
	\end{subfigure}
	\caption{Effect of Power Control fraction $\eta$ on the SIR coverage Probability, $\left(\lambda_M=2, \lambda_F=12, \alpha=3, N_M=12, N_F=4, B=1 \right)$.}
\label{sir_eta}
\vspace{-1.7em}
\end{figure*}
Fig. \ref{sir_eta} shows the effect of $\eta$ on $\mathrm{SIR}$ coverage probability when the cell association is based on maximum downlink received power and when it is based on DUDe. It can be observed that power control affects the cell-centered (corresponds to large SIR threshold) and cell-edged (corresponds to small SIR threshold) UEs differently, i.e., the centered UEs coverage decreases with power control, whereas the cell-edged UEs coverage increases with the middle value of $\eta=0.5$ and with full channel inversion $\left(\eta=1\right)$ it decreases. With $\eta=1$ the interference power become significant and hence decreases the overall coverage, therefore, $\eta$ should be optimized accordingly. Furthermore, comparing Fig. \ref{PC_NO_DUDe} and Fig. \ref{PC_DUDe} reveals that the effect of power control is more prominent when the association scheme is No-DUDe. This is due to the large cell size of the MBSs in the No-DUDe association as compared to the cell size of the MBSs in the DUDe association. 
%Similarly, in Fig. \ref{rate_coverage_comp} the rate coverage probability obtained from simulations and analysis is compared for different parameters. The mean load approximation \eqref{average_load} is used in these plots. From both Fig. \ref{rate_coverage_comp}a and Fig. \ref{rate_coverage_comp}b it is clear that both simulation and analysis match with each other. 
%\begin{figure} 
	%\centering
		%\includegraphics{Femto_tier_alpha3_MBias10_MBS20simulation_analytical.eps}
		%\caption{Femto-tier UL SINR coverage probability simulations vs Analytical, $\alpha=3, B=0.1, N_M=20, \lambda_F=12, \lambda_M=2, \lambda_u=3000$}
	%\label{sinr_cov_femto}
%\end{figure}
%%%%%%%%%
%%%%%%%%%%
%%%%%%%%%%%%%%%%%%%%%
%%%%%%%%%%%%%%%%%%%%%%%%
\begin{figure*} 
			\begin{subfigure}{.45\textwidth}
			\centering
			\includegraphics[width=1\linewidth]{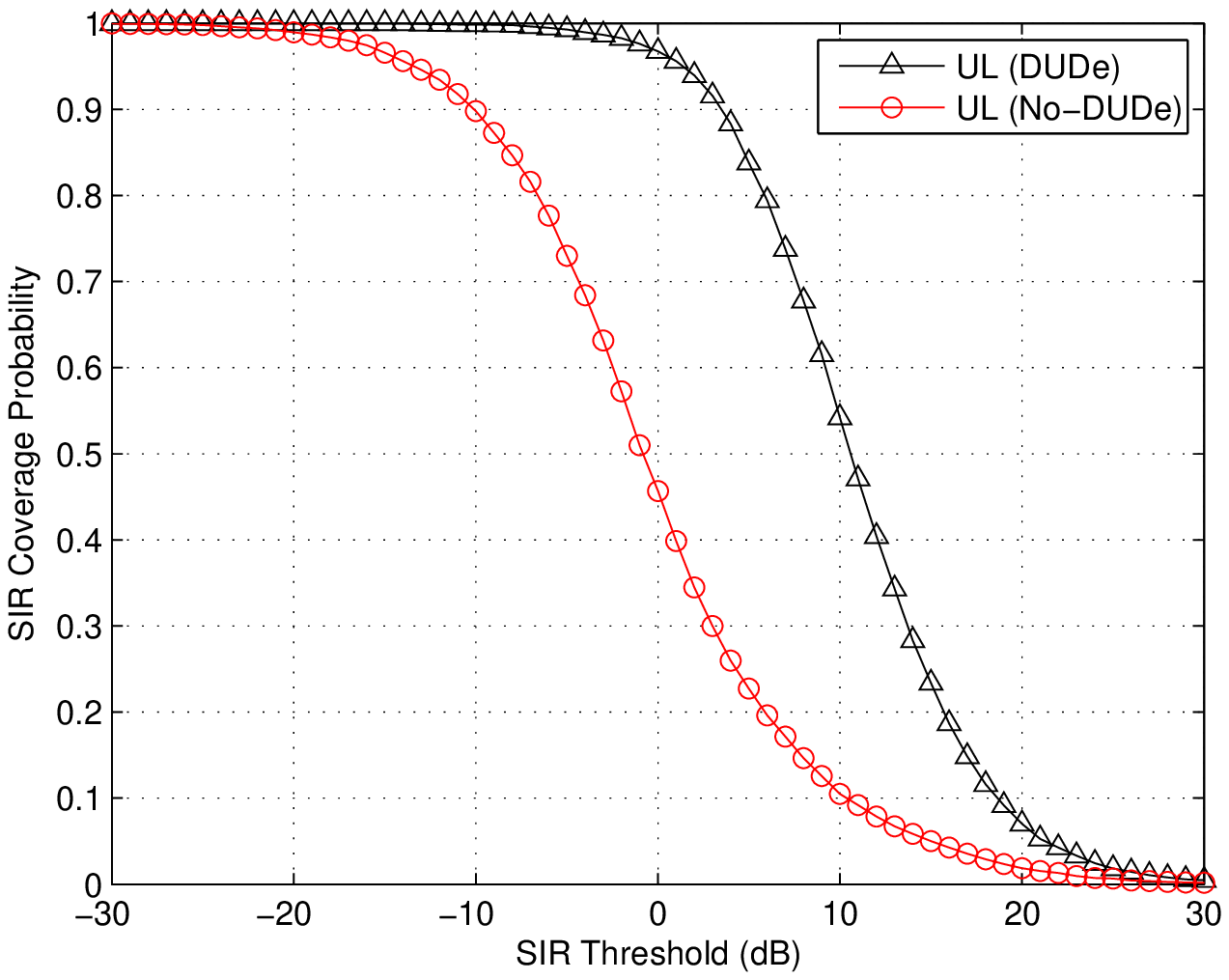}
			\caption{$N_M=12, N_F=12$}
			\label{NM_12NF_12E}
		\end{subfigure}\hfill
	\begin{subfigure}{.45\textwidth}
		\centering
		\includegraphics[width=1\linewidth]{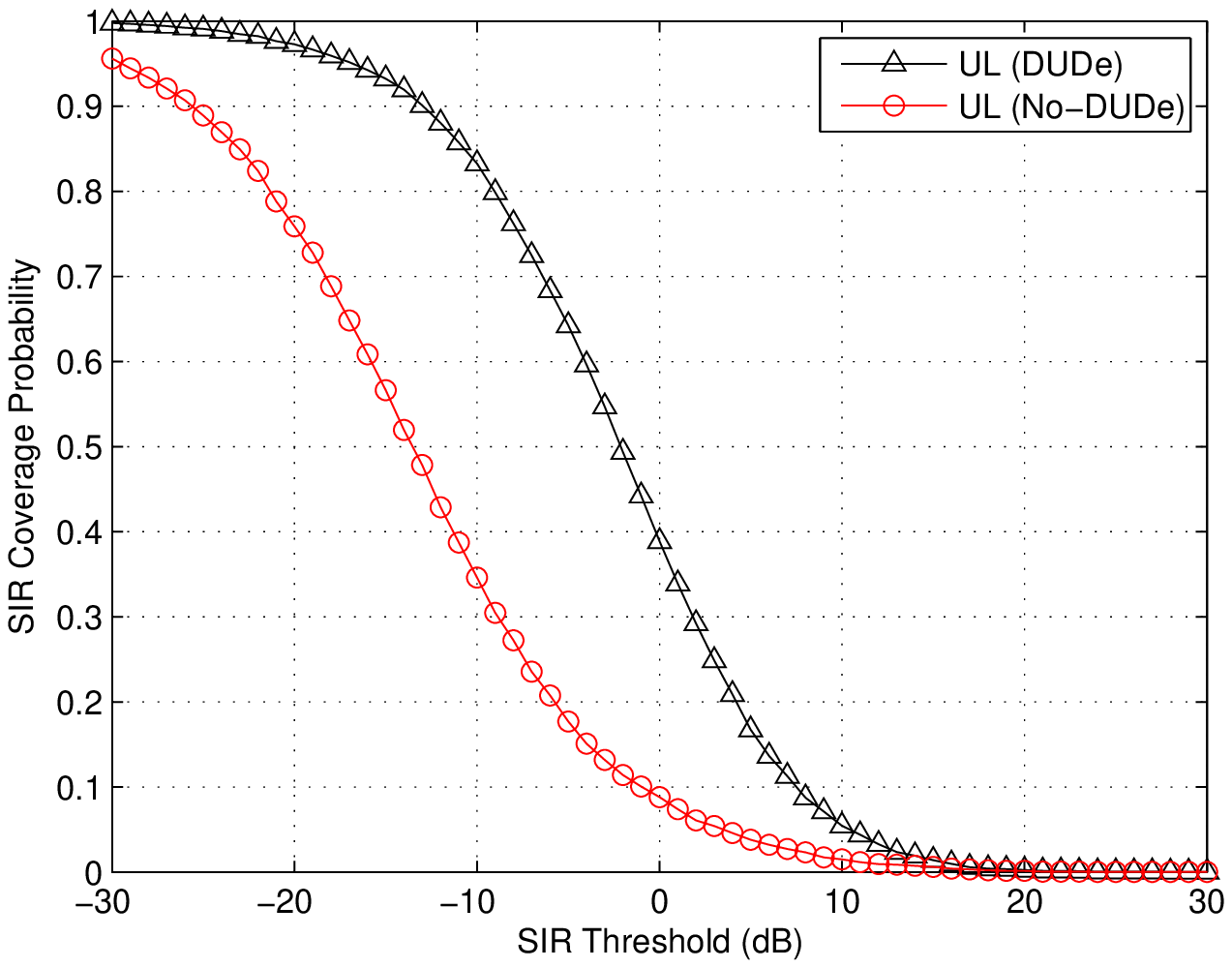}
		 \caption{$N_M=1, N_F=1$}
		 \label{NM_1NF_1E}
	\end{subfigure}\hfill
	\begin{subfigure}{.45\textwidth}
		\centering
			\includegraphics[width=1\linewidth]{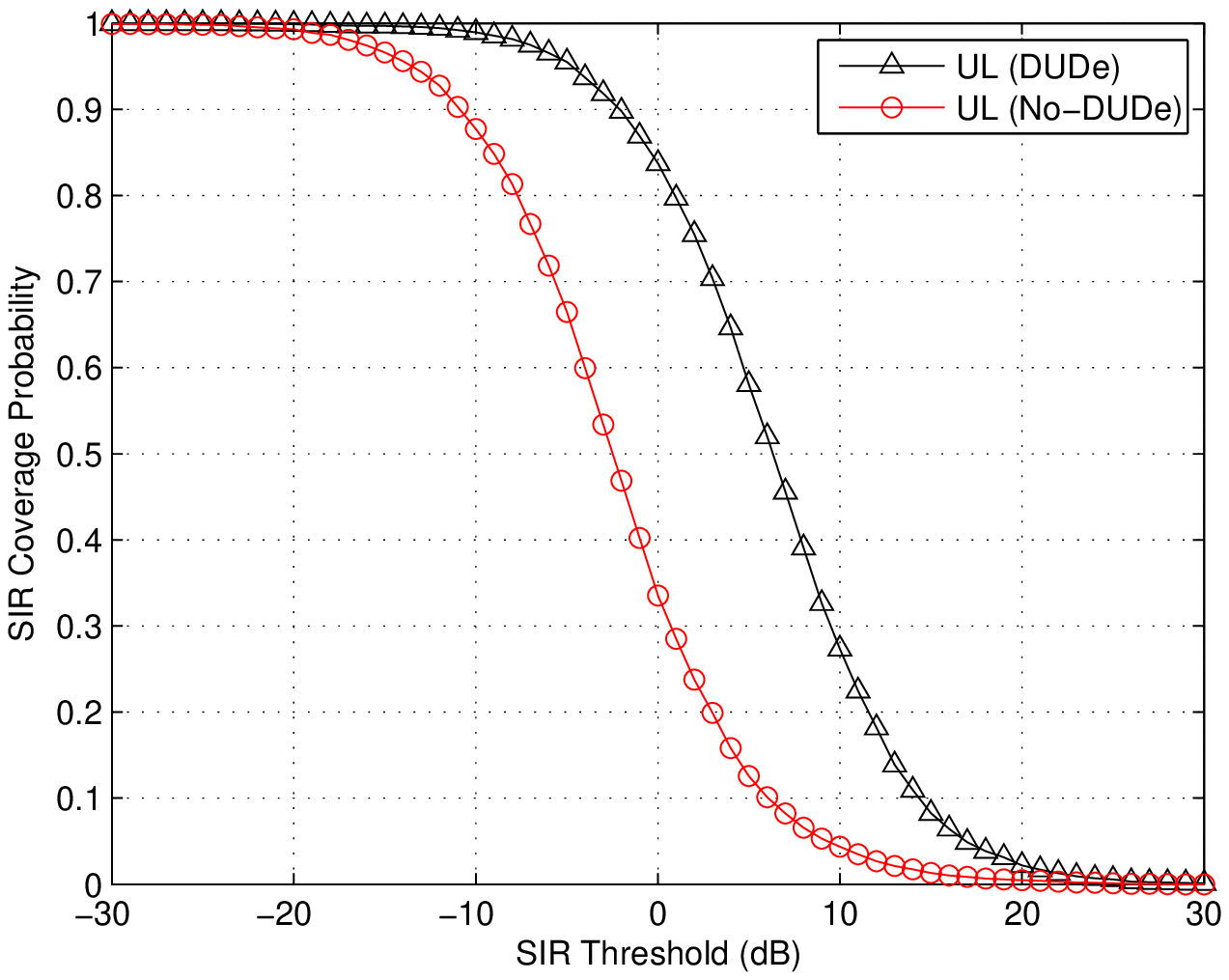}
			\caption{$N_M=12, N_F=4$}
			\label{NM_12NF_4E}
		\end{subfigure}\hfill
	\begin{subfigure}{.45\textwidth}
		\centering
		\includegraphics[width=1\linewidth]{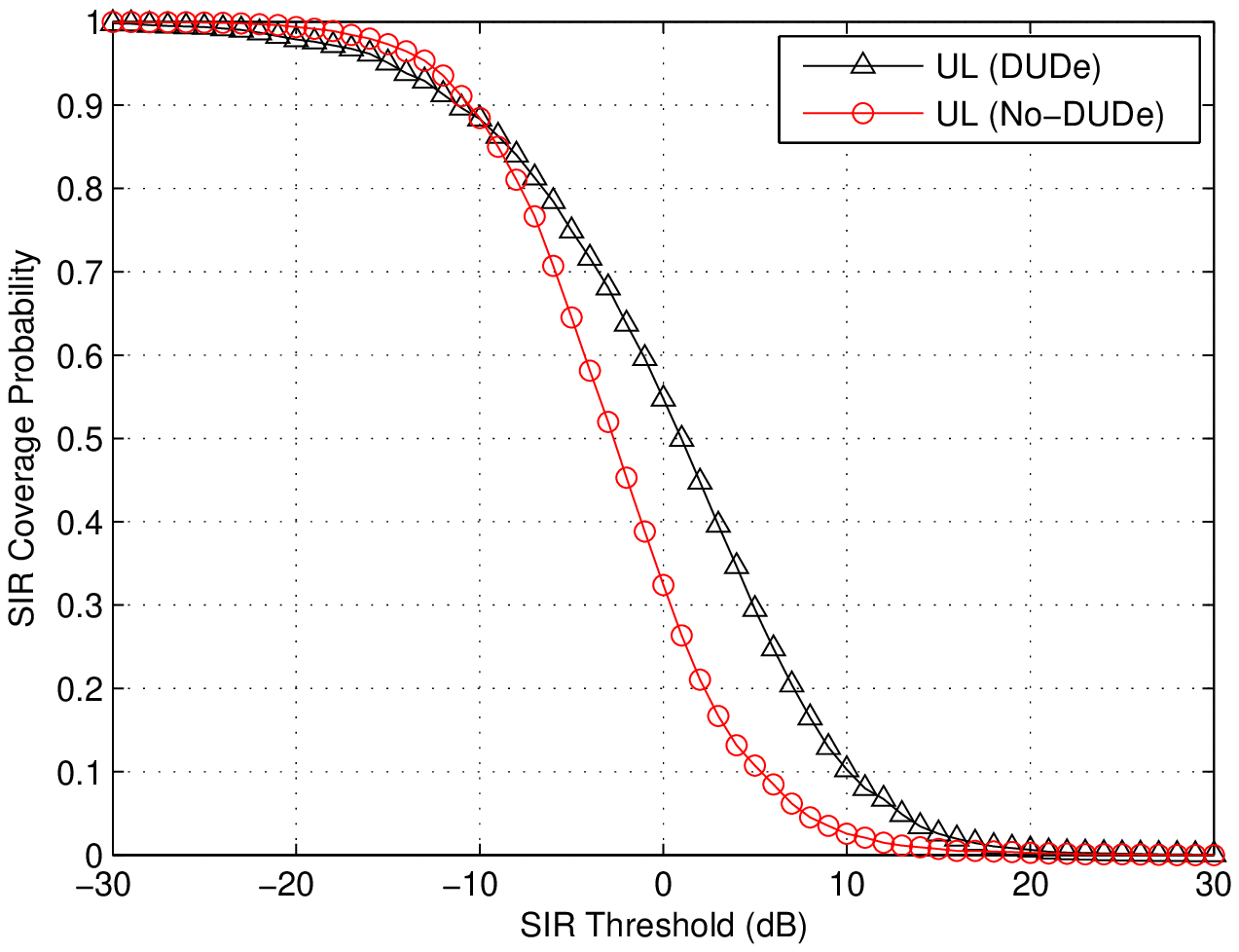}
		 \caption{$N_M=12, N_F=1$}
		 \label{NM_12NF_1E}
	\end{subfigure}\hfill
	\caption{Beamforming gain effect on the DUDe gain in term of SIR coverage probability with power control, $\left(\eta=0.5, \lambda_M=2, \lambda_F=12, \alpha=3, B=1\right)$. }
\label{DUDe_gain_pc}
\vspace{-1.7em}
\end{figure*}
\begin{figure*} 
			\begin{subfigure}{.45\textwidth}
			\centering
			\includegraphics[width=1\linewidth]{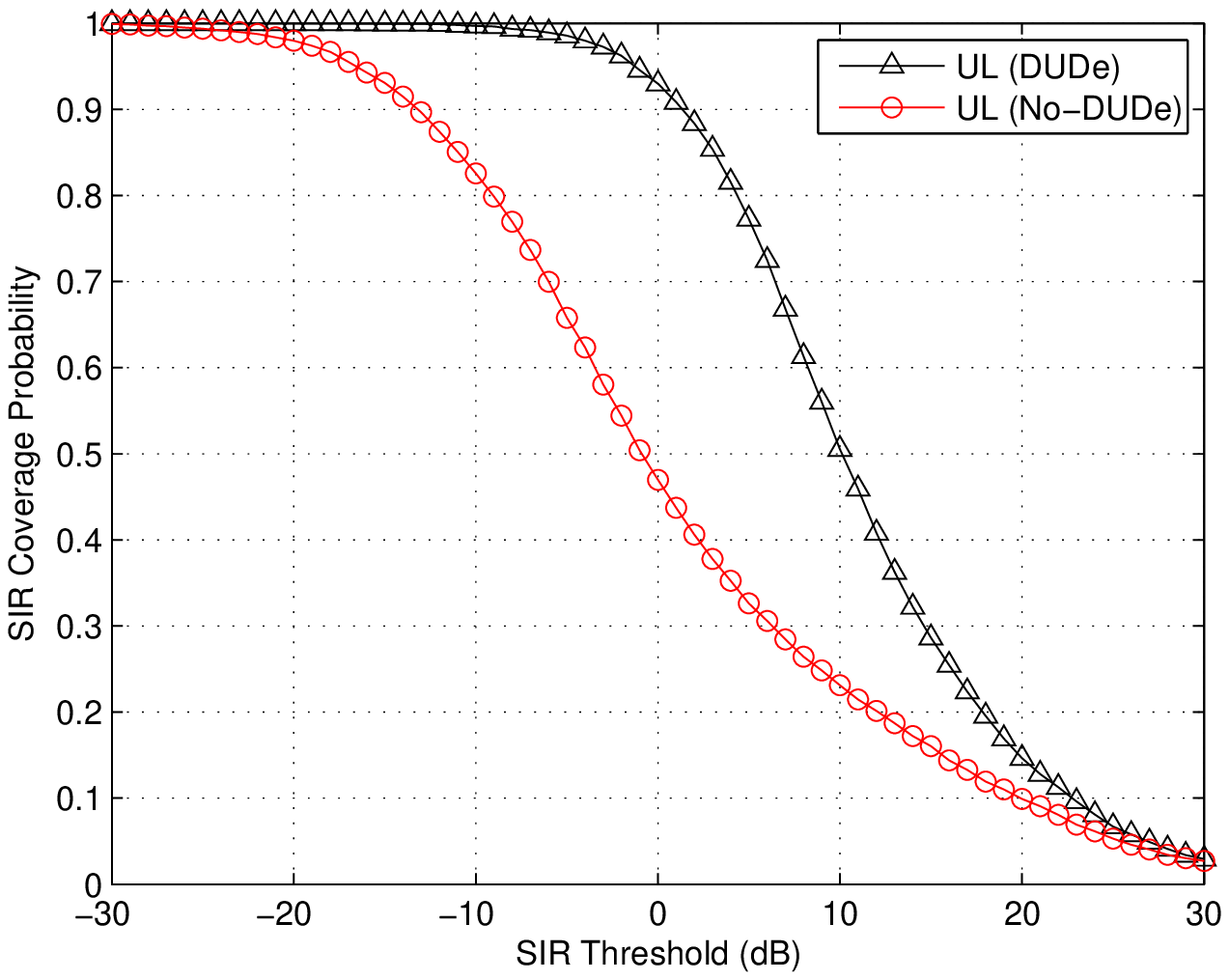}
			\caption{$N_M=12, N_F=12$}
			\label{NM_12_NF_12}
		\end{subfigure}\hfill
	\begin{subfigure}{.45\textwidth}
		\centering
		\includegraphics[width=1\linewidth]{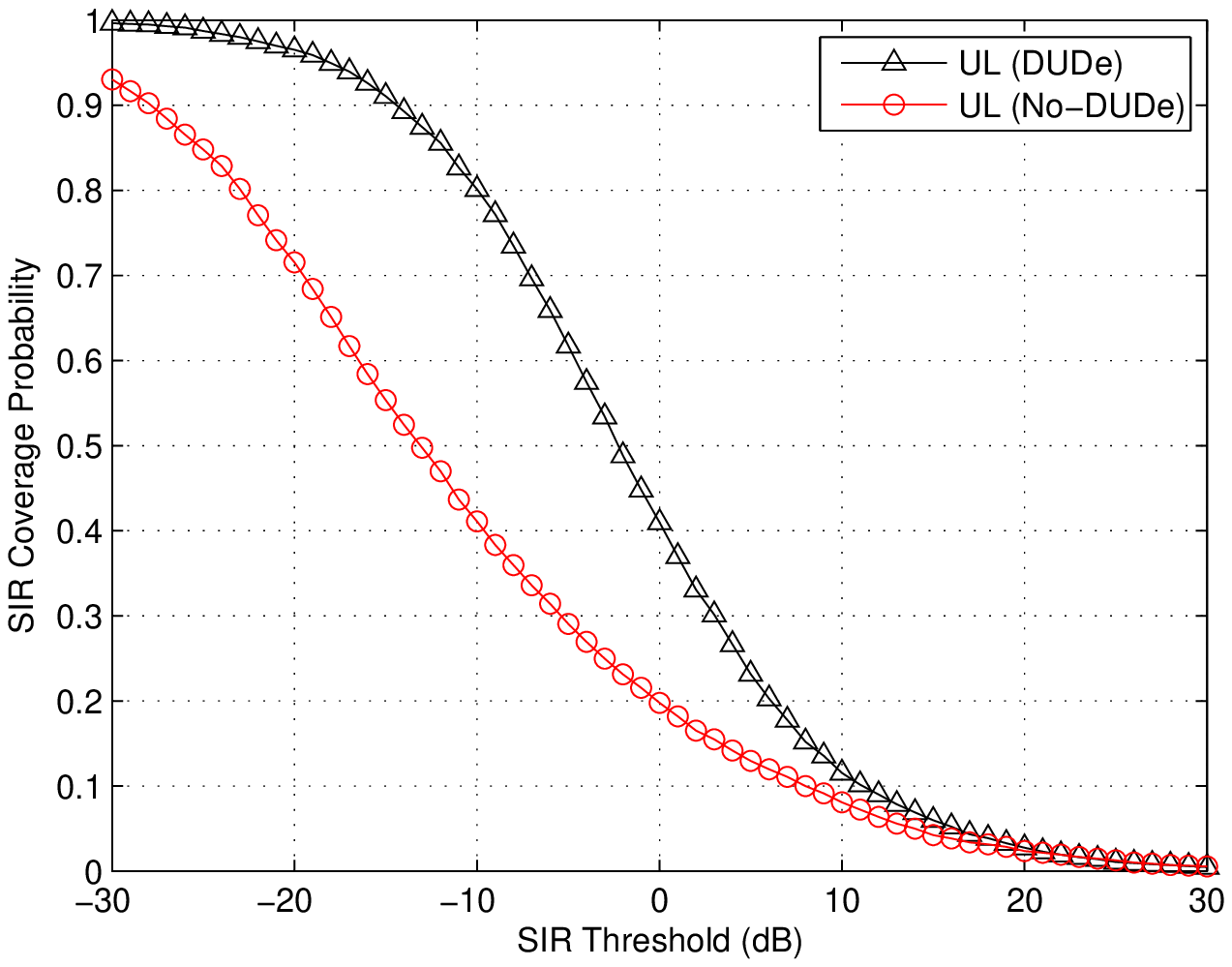}
		 \caption{$N_M=1, N_F=1$}
		 \label{NM_1_NF_1}
	\end{subfigure}\hfill
	\begin{subfigure}{.45\textwidth}
		\centering
			\includegraphics[width=1\linewidth]{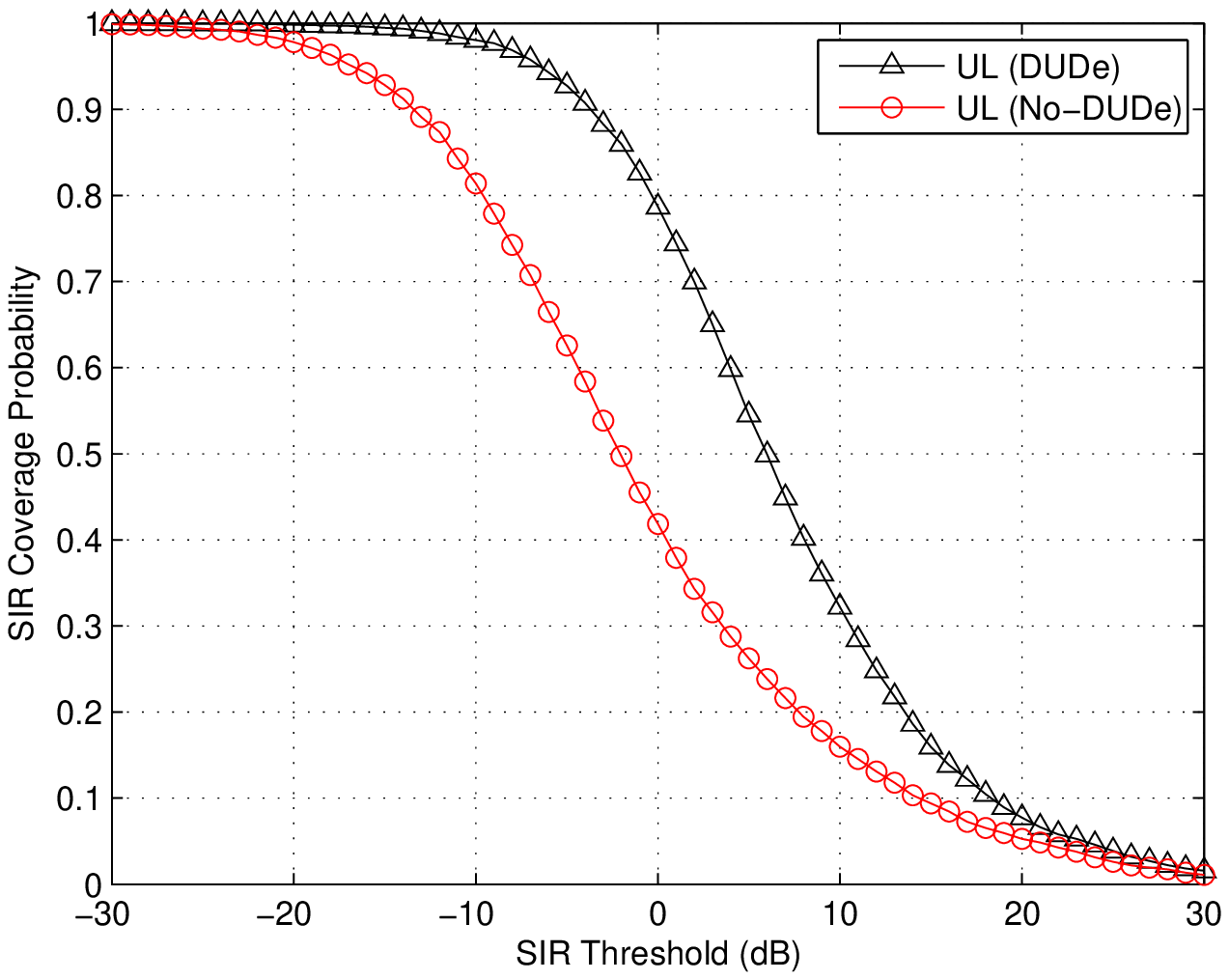}
			\caption{$N_M=12, N_F=4$}
			\label{NM_12_NF_4}
		\end{subfigure}\hfill
	\begin{subfigure}{.45\textwidth}
		\centering
		\includegraphics[width=1\linewidth]{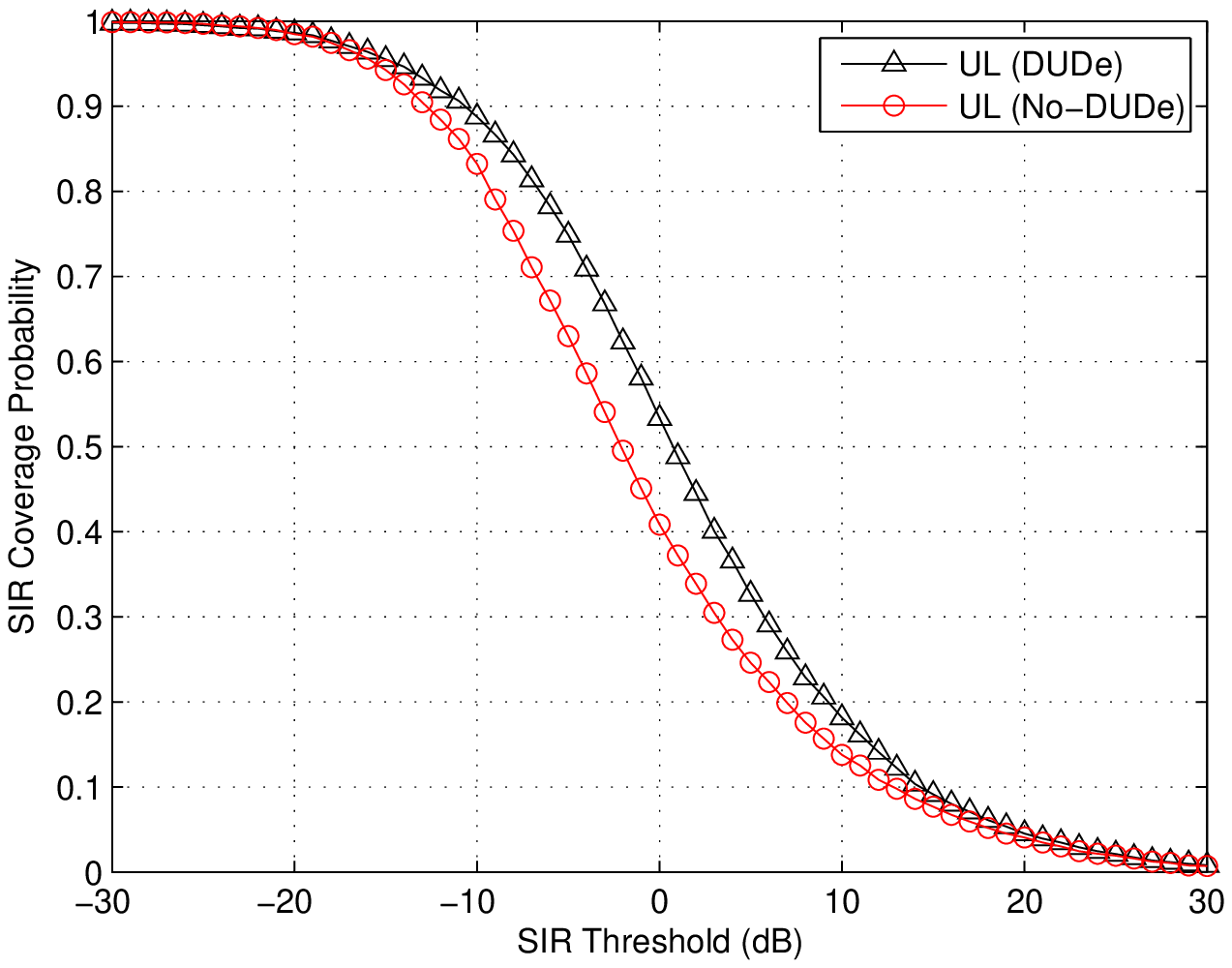}
		 \caption{$N_M=12, N_F=1$}
		 \label{NM_12_NF_1}
	\end{subfigure}\hfill
	\caption{Beamforming gain effect on the DUDe gain in term of SIR coverage probability without power control, $\left(\eta=0, \lambda_M=2, \lambda_F=12, \alpha=3, B=1 \right)$. }
\label{DUDe_gain_nopc}
\vspace{-1.7em}
\end{figure*}
%\begin{figure} 
			%\begin{subfigure}{.5\textwidth}
			%\centering
			%\includegraphics[width=1\linewidth]{SIR_coverage_MBS1_1.eps}
			%\caption{$N_M=1$}
			%\label{NM_1}
		%\end{subfigure}
	%\begin{subfigure}{.5\textwidth}
		%\centering
		%\includegraphics[width=1\linewidth]{SIR_coverage_MBS5_2.eps}
		 %\caption{$N_M=5$}
		 %\label{NM_5}
	%\end{subfigure}
	%\begin{subfigure}{.5\textwidth}
		%\centering
			%\includegraphics[width=1\linewidth]{SIR_coverage_MBS10_2.eps}
			%\caption{$N_M=10$}
			%\label{NM_10}
		%\end{subfigure}
	%\begin{subfigure}{.5\textwidth}
		%\centering
		%\includegraphics[width=1\linewidth]{SIR_coverage_MBS15_3.eps}
		 %\caption{$N_M=15$}
		 %\label{NM_15}
	%\end{subfigure}
	%\caption{Effect of increasing MBS' antennas, $N_M$, on DUDe gain in term of SIR coverage, $\left(\lambda_M=2, \lambda_F=12 \alpha=3, B=1 \right)$. }
%\label{DUDe_gain}
%\end{figure}
%\begin{figure}
	%\centering
		%\includegraphics{H:/Report/rate_coverage_vs_biasing1.eps}
	%\label{fig:rate_coverage_vs_biasing1}
%\end{figure}
 %Whereas, at high SINR threshold the simulation curve is slightly above the numerical and this is due to the assumption that the interfering user is at the same distance as the tagged user.% \textcolor{blue}{ Furthermore, it can be noticed that the coverage probability for non biasing case is greater than the biasing case. The SINR coverage is always greater in the non biasing case regardless of biasing towards femto-tier or macro-tier.}

Fig. \ref{DUDe_gain_pc} shows how the gain provided by the DUDe association over No-DUDe association in term of SIR coverage probability changes with the beamforming gain of both tiers. It is important to mention that the UL coverage probability of the network when the association is based on maximum DL received power averaged over fading can be derived by similar tools and methods used in this paper. It is clear from the figure that the gain of DUDe association over No-DUDe is maximum when both tiers have the same beamforming gain and decreases otherwise. When $N_M$ is large compared to $N_F$, the beamforming gain provided by a MBS increases, which enlarges the association region of a MBS. As a result of which UEs closer to the FBSs become associated with MBSs. These boundary UEs, which are connected to macro-tier, create strong interference at nearby FBSs when they transmit to their serving MBSs. Whereas, when both tiers have the same beamforming gain then the coverage region of both tiers are the same and the interference created by the boundary UEs is not that strong. Thus, the DUDe gain over No-DUDe is high when both tier have the same beamforming. 
%Although these UEs have strong signal connection with their serving MBSs, they create strong interference at nearby FBSs, which reduces the gain of DUDe. 
In other words, we can say that as the difference in beamforming gain of both tiers increases, the gain provided by the DUDe over No-DUDe decreases. Fig. \ref{DUDe_gain_nopc} shows the same effect when UL power control is not utilized. 
%In No-DUDe case there is a huge difference in the transmit power of the MBS and FBS in addition to the difference in the beamforming gain, whereas in the DUDe the same effect can be observed when the disparity in the number antennas of both tier is very high.
%Therefore, it can be inferred that when there is a high asymmetry between MBS' and FBS' antennas then the decoupling gain  would not be that high as compared to SISO network. %By taking such observation, and considering extreme scenario such as Massive MIMO MBSs,  it can be concluded that DUDe can not increase  Massive MIMO network's coverage.
%\begin{comment}
%\begin{figure} 
	%\centering
		%\includegraphics[scale=0.7]{Rate_coverage_dif_NM.eps}
		%\caption{Effect of number of MBS antennas and biasing on rate coverage, $\left(\lambda_M=3, \lambda_F=18, \lambda_u=3000, \alpha=3, W=10\text{MHz}\right)$}
	%\label{ratecov_4_diff_NM}
%\end{figure}
%\end{comment}

Fig. \ref{ratecov_4_diff_NM} shows the effect of the number of MBS's antennas and biasing on rate coverage probability. For no biasing case $B=1$, increasing $N_M$ from $1$ to $20$ decreases the rate coverage.  To explain this effect, we know that the rate coverage depends on the load on a BS \eqref{rate_cov_basic}. When $N_M$ is high, the coverage region of macro-tier increases and most of the UEs become attached to MBSs due to which the macro-tier is overloaded. Thus the overall rate coverage probability drops. Further, we can see from the figure that when $N_M=1$, then no-biasing gives us the maximum rate coverage, which is in accordance with the result of \cite{Geff_DUDE}. However, for higher $N_M$ we see that biasing improves the rate coverage. From the network design perspective, we see that increasing $N_M$ can degrade the rate coverage, therefore, to benefit from a large number of MBSs' antennas we need a suitable biasing towards femto-tier.

Fig. \ref{effect_alpha_lambdaf} illustrates the effect of FBSs' density and path-loss exponent $\alpha$ on the rate coverage probability for the association scheme of DUDe and No-DUDe. It can be observed that by changing $\alpha$ from $3$ to $4$ increases the rate coverage probability for both DUDe and No-DUDe, which comes from the decrease in the interference power. It can be further observed that an increase in $\lambda_F$ increases the rate coverage for the DUDe case. This improvement in the rate coverage comes from the inherent property of the DUDe to better handle interference. On the other hand, for No-DUDe association scheme, increasing $\lambda_F$ slightly improves the rate coverage for  centered UEs (large rate threshold) while decreases the rate coverage of cell-edged UEs (small rate threshold). When $\lambda_F$ increases then the load on BS decreases due to which  the rate coverage improves for the cell-centered users. However, with the increase in $\lambda_F$, the cell size of a BS decreases and by using channel inversion the cell-edged UEs transmit power also reduces, thus the coverage of cell-edge UEs reduces.

%%%%%%%%%%%%%%%%%%%%%%%%%
\begin{figure} 
	\centering
		\includegraphics[scale=0.54]{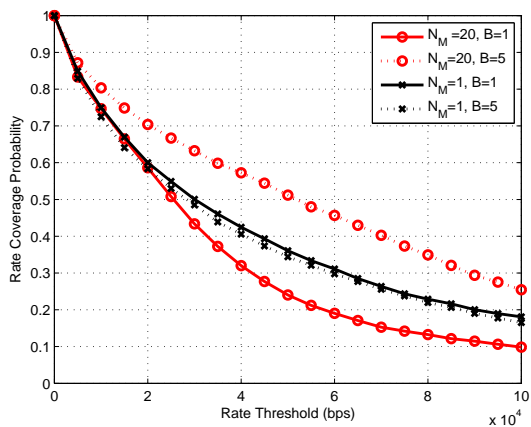}
		\caption{Effect of number of MBS antennas and biasing on rate coverage, $\left(\lambda_M=3, \lambda_F=18, \lambda_u=3000, \alpha=3,\eta=1\right)$}
	\label{ratecov_4_diff_NM}
	\vspace{-1.7em}
\end{figure}
\subsection{Optimal bias and optimal power control fraction}
%\begin{figure}
	%\centering
		%\includegraphics[scale=0.7]{SIR_coverage_vs_bias.eps}
		%\caption{Optimal bias for SIR coverage $\left(N_M=8, \lambda_M=2, \lambda_F=14, \alpha=3.5 \right)$}
		%\label{optimalB_sir}
%\end{figure}
Fig. \ref{optimal_B_sir} shows the effect of biasing on $\mathrm{SIR}$ coverage probability for $\eta=0$ and $\eta=1$. For $\eta=0$ the optimal coverage probability is given by no biasing i.e., $B=\frac{B_F}{B_M}=1$ or $B=0$dB as shown by Fig. \ref{no_pc}. The SIR is independent of the load and depends on the density of BSs, path-loss, beamforming gain of the BSs, and the SIR threshold $\tau$, and when $\eta=0$ then all UEs transmit with the same power. By using biasing we force a UE to associate to a BS to which the UE connection is not strong and thus the SIR coverage probability reduces.
However, from Fig. \ref{with_pc} we see that when $\eta=1$ the optimal SIR is given by  $B=5$dB. With power control the transmit power of a UE is proportional to its distance from the BS and the transmit power of the UEs at the cell-edged is greater than the cell-centered UEs. Further, when the beamforming gain $N_M$ of the macro-tier is greater than the femto-tier then cell-edge UEs of macro cells transmit with large power and generate high interference. Therefore, offloading these cell-edged UEs to femto-tier improves the SIR coverage.    
\begin{figure*} 
			\begin{subfigure}{.45\textwidth}
			\centering
			\includegraphics[width=1\linewidth]{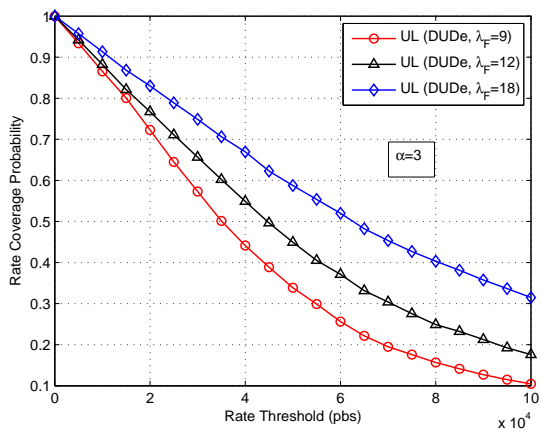}
			\caption{$\alpha=3$}
			\label{alpha3_DUDe}
		\end{subfigure}\hfill
	\begin{subfigure}{.45\textwidth}
		\centering
		\includegraphics[width=1\linewidth]{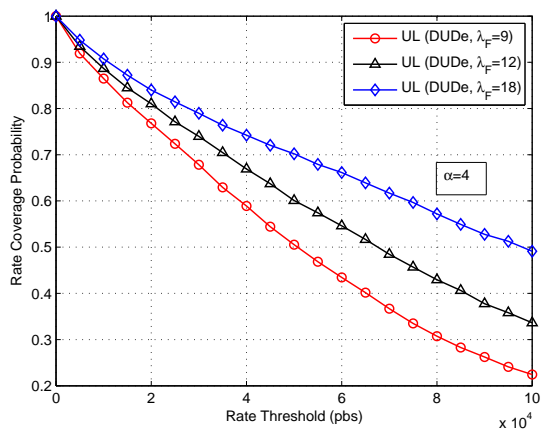}
		 \caption{$\alpha=4$}
		 \label{alpha4_DUDe}
	\end{subfigure}\hfill
	\begin{subfigure}{.45\textwidth}
		\centering
			\includegraphics[width=1\linewidth]{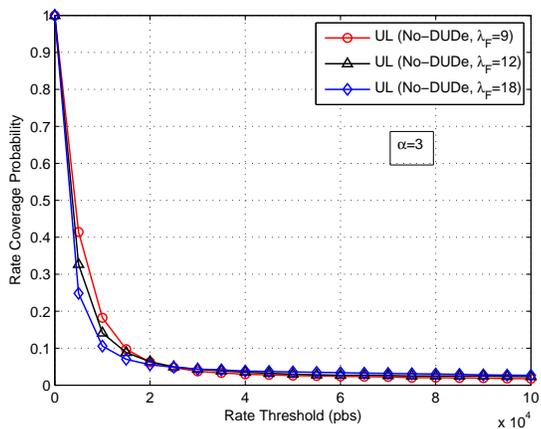}
			\caption{$\alpha=3$}
			\label{alpha3_No_DUDe}
		\end{subfigure}\hfill
	\begin{subfigure}{.45\textwidth}
		\centering
		\includegraphics[width=1\linewidth]{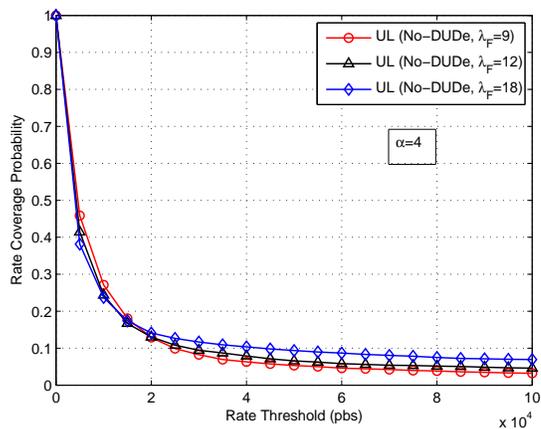}
		 \caption{$\alpha=4$}
		 \label{alpha4_No_DUDe}
	\end{subfigure}\hfill
	\caption{Effect of $\lambda_F$ and $\alpha$ on the rate coverage for DUDe and No-DUDe association $\left(\eta=1, \lambda_M=3, N_M=6, N_F=2, B=5, \lambda_U=3000 \right)$. }
\label{effect_alpha_lambdaf}
\vspace{-1.7em}
\end{figure*}
\begin{figure*} 
			\begin{subfigure}{.45\textwidth}
			\centering
			\includegraphics[width=1\linewidth]{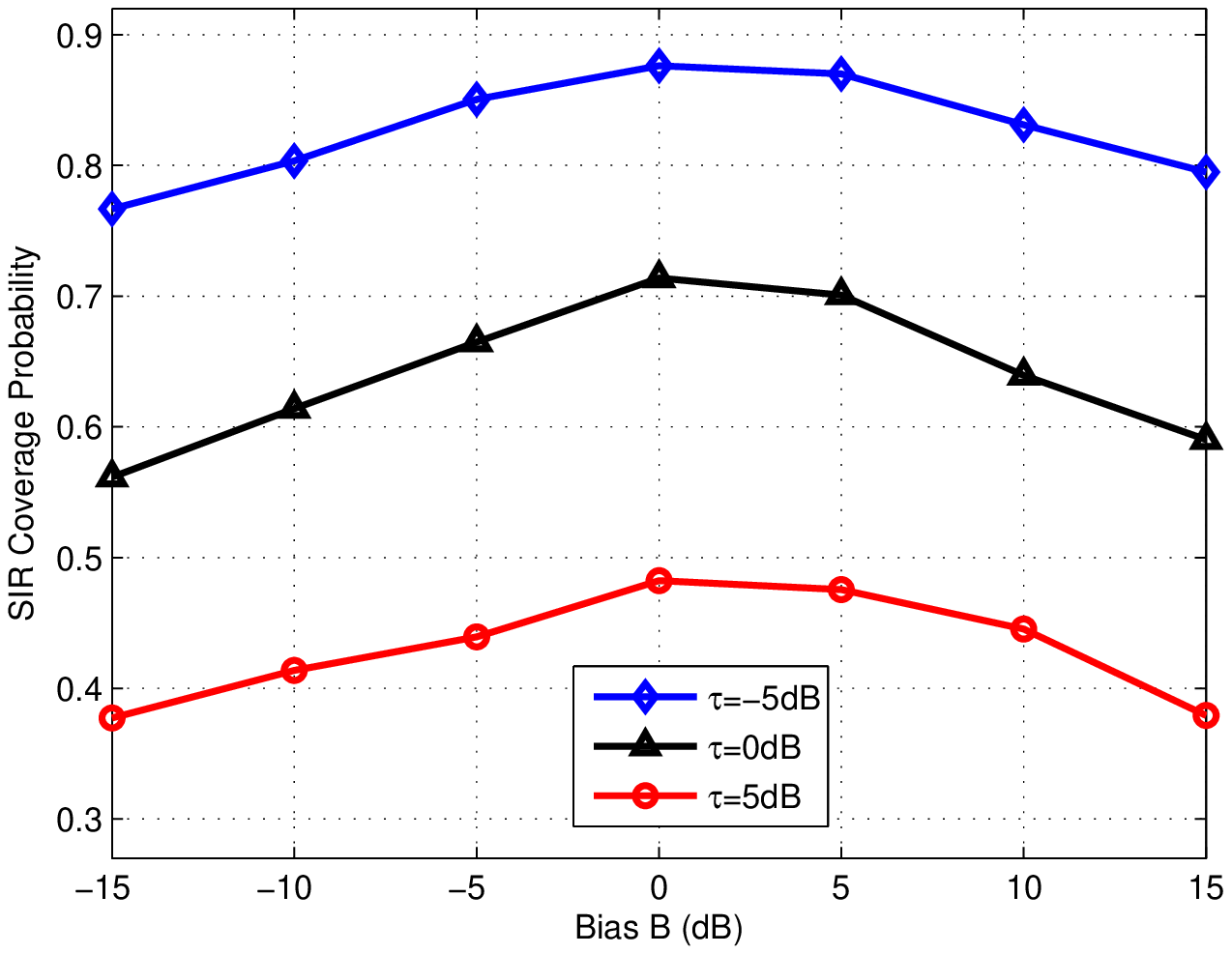}
			\caption{$\eta=0$}
			\label{no_pc}
		\end{subfigure}\hfill
	\begin{subfigure}{.45\textwidth}
		\centering
		\includegraphics[width=1\linewidth]{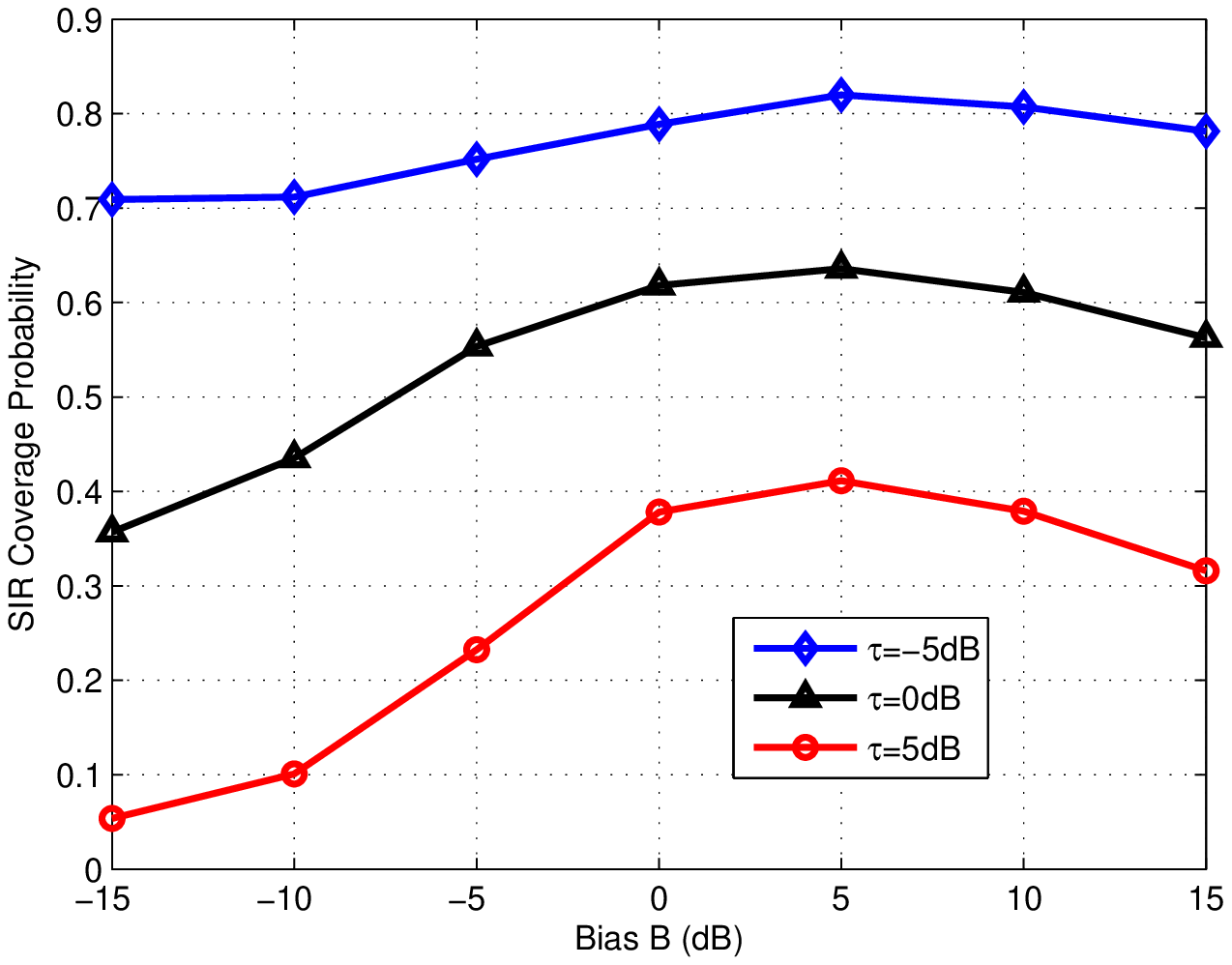}
		 \caption{$\eta=1$}
		 \label{with_pc}
	\end{subfigure}\hfill
\caption{Optimal bias for SIR coverage $\left(N_M=20, N_F=2, \lambda_M=2, \lambda_F=10, \lambda_U=3000, \alpha=3\right)$}
		\label{optimal_B_sir}
		\vspace{-1.7em}
\end{figure*}

\begin{figure*} 
			\begin{subfigure}{.45\textwidth}
			\centering
			\includegraphics[width=1\linewidth]{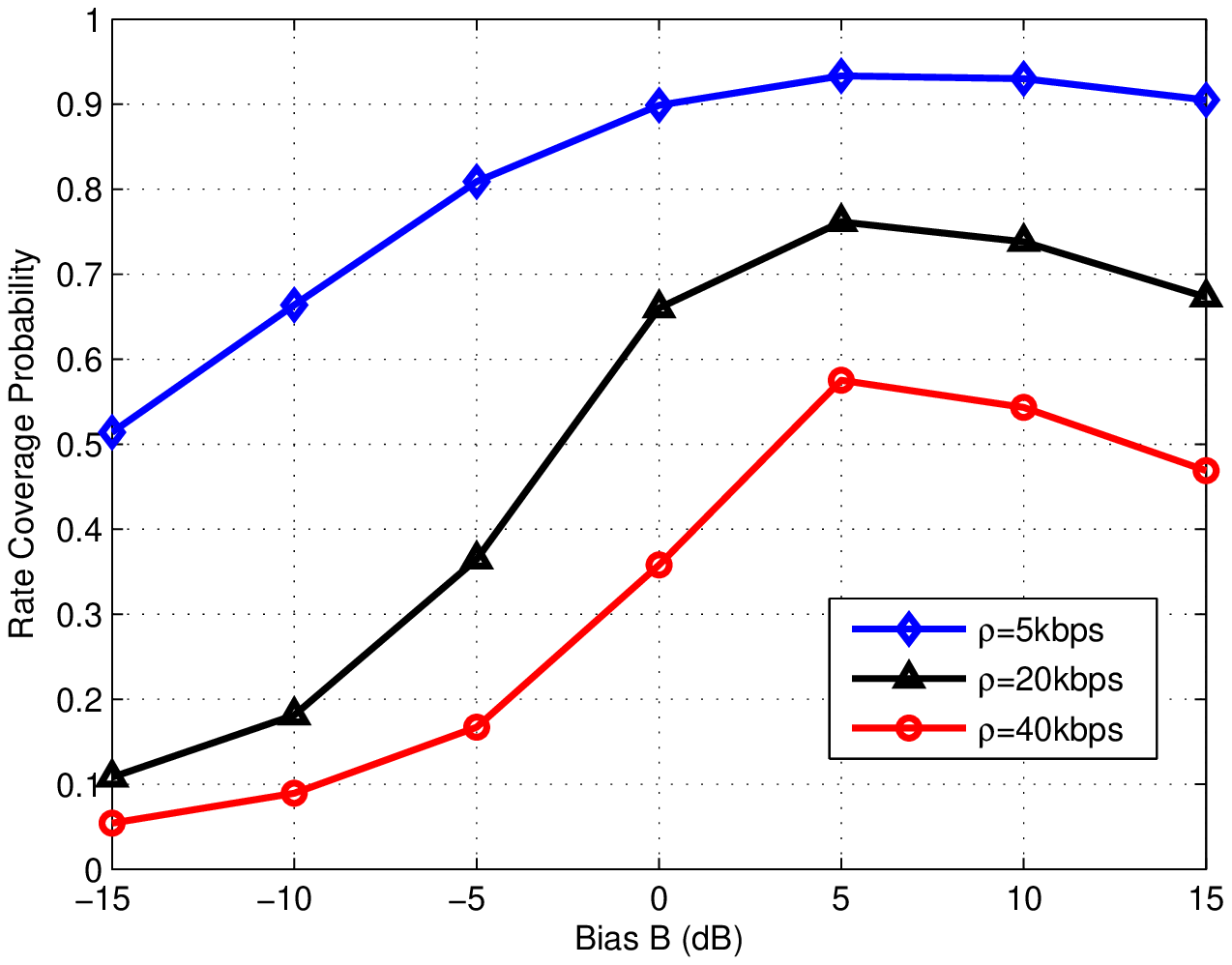}
			\caption{}
			\label{low_thr}
		\end{subfigure}\hfill
	\begin{subfigure}{.45\textwidth}
		\centering
		\includegraphics[width=1\linewidth]{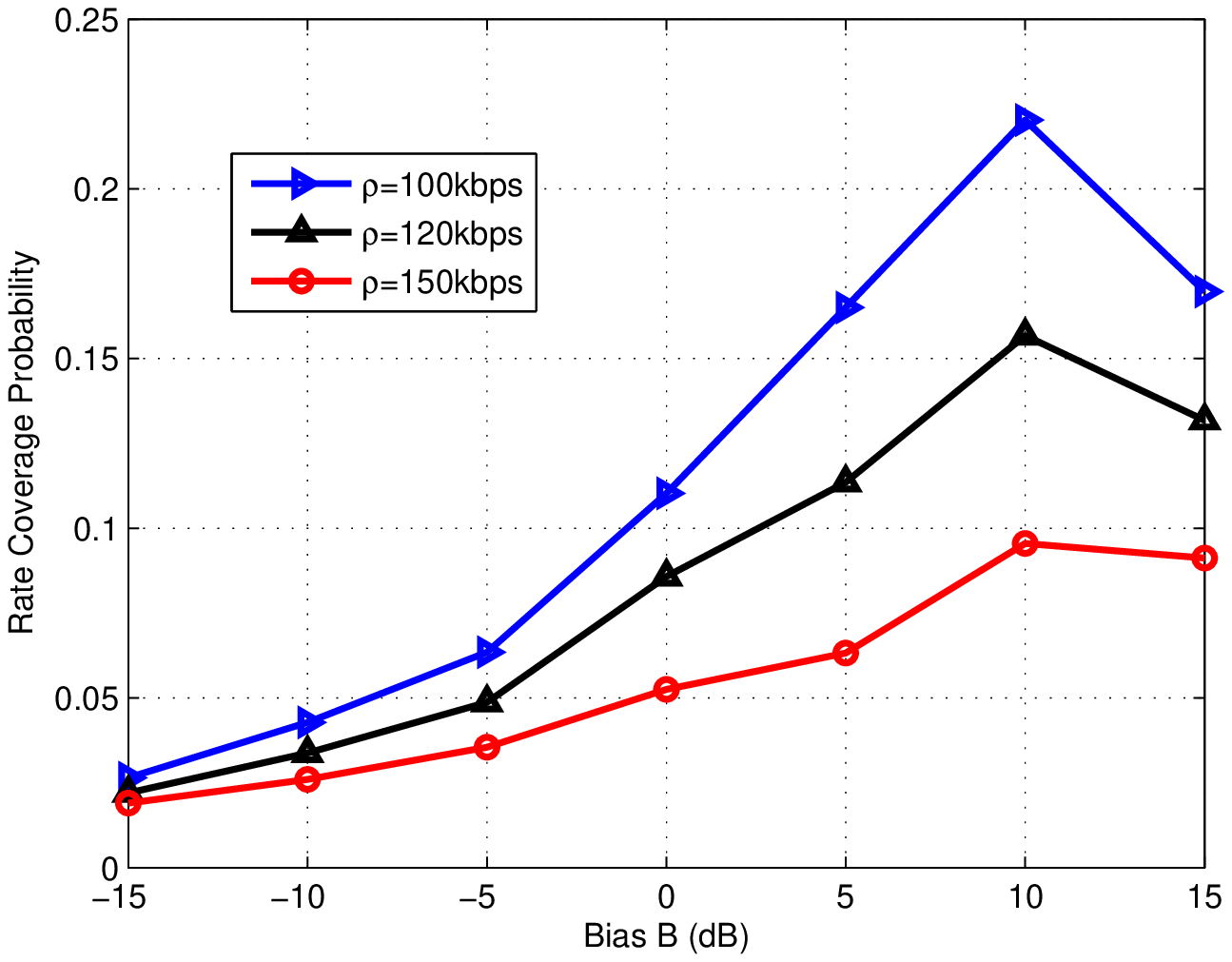}
		 \caption{}
		 \label{high_thr}
	\end{subfigure}\hfill
\caption{Optimal bias for rate coverage $\left(\eta=1, N_M=20, N_F=2, \lambda_M=2, \lambda_F=10, \lambda_U=3000, \alpha=3\right)$}
		\label{optimal_B_rate}
		\vspace{-1.7em}
\end{figure*}
%\begin{figure}
	%\centering
		%\includegraphics[scale=1]{rate_coverage_biasing.eps}
		%\includegraphics[scale=0.7]{rate_coverage_vs_biasing1.eps}
		%\caption{Optimal bias for rate coverage $\left(N_M=20, \lambda_M=3, \lambda_F=18, \lambda_U=1000, \alpha=3, W=1\right)$}
		%\label{optimalB_rate}
%\end{figure}
The rate depends on the load and using appropriate value of biasing can maximize the rate coverage. To find the closed form expression for the optimal bias is too challenging in our system model. However, the optimal value can be found by a linear search. Fig. \ref{optimal_B_rate} shows the rate coverage against biasing for different rate threshold $\rho$. It is clear from the figure that the maximum rate coverage is given by offloading UEs towards femto-tier. However, this optimal bias value changes with $\rho$. When $\rho$ is small (corresponds to cell-edged UEs) then we need a small value of $B$ whereas for large $\rho$ (corresponds to cell-centered UEs) then we need more aggressive biasing as shown in Fig. \ref{low_thr} and Fig. \ref{high_thr}, respectively. One can observe  that for $B<0$ dB the rate coverage is very low. When the beamforming gain of the macro-tier is high, the coverage region is also large as compared to femto-tier and biasing towards macro-tier further increases the coverage region of MBSs  (see Fig. \ref{biasing}). Due to this enlargement of the coverage region, a large number of UEs becomes attached to the macro-tier and it becomes overloaded, which drops the rate coverage probability. In \cite{Geff_DUDE} it is shown that for SISO network the UL rate coverage is maximized when the association is based on minimum path-loss. However, for MIMO setup this is not the case. Comparing the UL offloading with the DL one can see that in the DL we need more aggressive offloading of UEs to the small cell, because there is a high disparity in both the transmit powers and beamforming gains of macro and femto BSs. Whereas, in the uplink the load imbalance is only due to the difference in the beamforming gain of the macro and femto BSs.

%Case 2. DLAP = Macro-BS and ULAP = Femto-BS: 
	%\label{caseb3}
	%\end{equation}
	%Here, the intersection of the event is give by $P_{M}N_M X_M^{-\alpha} \leq P_{F} X_F^{-\alpha}$, and the association probability is 
	%\begin{multline}
	%Pr\left(Case 3\right)=Pr\left( X_M \geq B X_F\right) = \int_{0}^{\infty}\left(1-F_{X_M}\left(B X_F\right)\right)f_{X_F}\left(X_F\right)d_{X_F} \\
	%= \frac{\lambda_F}{\lambda_F+\left(\frac{P_M N_M}{P_F}\right)^{2/\alpha}\lambda_M}  
	%\label{prob_case3},
	%\end{multline}
	%where $B=\left(\frac{P_M N_M}{P_F}\right)^{2/\alpha}$ is the same as defined above. Similarly, we use the name of femto users when we want to refer to the users fall in Case 3.
	%Case 4. DLAP = Femto-BS and ULAP = Macro-BS: As stated earlier that the intersection of this event is zero so its probability remains zero 
%\begin{equation}
	%Pr\left\{\left\{P_{M}N_M X_M^{-\alpha} \leq P_{F}X_F^{-\alpha} \right\} \bigcap  \left\{N_M X_M^{-\alpha}>\gamma X_F^{-\alpha}\right\}\right\}
	%\label{caseb4}
	%\end{equation}
%Hence, the probability $Pr\left(Case 4\right)=0 $, we will not discuss this case anymore. 
Fig. \ref{optimal_E_rate} shows the rate coverage against $\eta$. The power control fraction $\eta$ affects the cell-edged and cell-centered UEs differently. For cell-edged UEs the optimal rate coverage is given by the median value of $\eta$ is shown in Fig. \ref{low_thrE}. Whereas for cell-centered UEs the optimal rate coverage is given by without uplink power control $\eta=0$ as shown in Fig. \ref{high_thrE}. Therefore, based on the target rate threshold, the appropriate value of $\eta$ can be chosen to optimize the rate coverage. 
\begin{figure*} 
			\begin{subfigure}{.45\textwidth}
			\centering
			\includegraphics[width=1\linewidth]{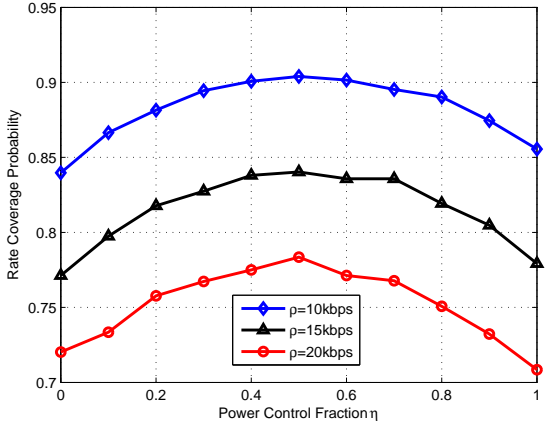}
			\caption{}
			\label{low_thrE}
		\end{subfigure}\hfill
	\begin{subfigure}{.45\textwidth}
		\centering
		\includegraphics[width=1\linewidth]{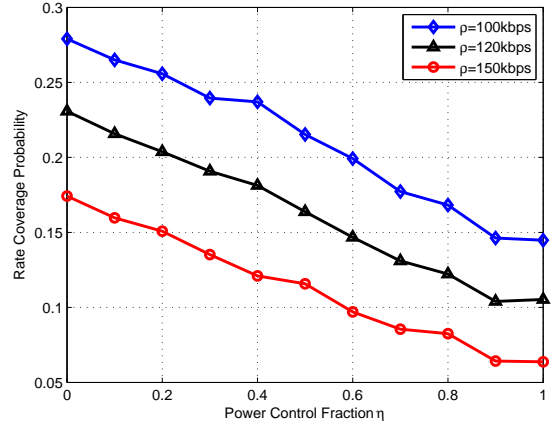}
		 \caption{}
		 \label{high_thrE}
	\end{subfigure}\hfill
\caption{Optimal $\eta$ for rate coverage $\left(N_M=4, N_F=2, \lambda_M=2, \lambda_F=10, \lambda_U=3000, \alpha=3\right)$}
		\label{optimal_E_rate}
		\vspace{-1.7em}
\end{figure*}
\section{Conclusion}
Using tools from stochastic geometry, the UL performance of a two-tier random network is studied, where the cell association is based on DL and UL decoupling. Multiple antennas are considered at BSs, and single antennas are considered at UEs. The position of the MBSs, FBSs, and UEs are modeled using a 2-D PPP. Maximal ratio combining has been used at the MBS and tractable analytical expressions have been derived for the rate and SIR coverage probability. It has been shown that the gain (in term of SIR coverage probability) of the decoupled DL and UL association over the coupled DL and UL association is maximum when both tiers have the same number of antennas (same beamforming gain). It has also been observed that in order to leverage the benefits of multiple antennas in DUDe network, offloading of UEs to small cell is required. %It has also been observed that without using biasing, high number of antennas at MBSs do not necessarily improve the rate coverage. However, biasing in SISO network does not improve network performance.
%The UL performance of two-tier random network is studied, where the cell association is based on DL and UL decoupling. Multiple antennas are considered at MBS and single antennas are considered at FBS. The position of the MBSs, FBSs, and user equipments are modeled using a 2D PPP. Maximal ratio combining has been used at the MBS and tractable analytical expressions have been derived for the rate and SIR coverage probability by using  stochastic geometry. It has been shown that when the number of antennas at MBS is large then macro-tier becomes overloaded and offloading some users to femto-tier improve the overall rate coverage of the network. High path-loss also give us better rate coverage probability. It has been illustrated that DUDe has the ability to handle interference more efficiently as compared to the association scheme based on DL received power.  
A future extension might consider to study the performance of both DL and UL for MIMO network, and to find the optimal offloading strategy, which jointly optimizes both the DL and the UL performance. To investigate the potential gain offered by using multiple antennas BSs and using interference cancellation would also be an interesting research direction. 
%\ifCLASSOPTIONcaptionsoff
  %\newpage
%\fi
%\newpage	
\appendices
\section{}
\textit{Proof of Lemma 1:} The association criterion when a typical UE connects to a MBS both in the UL and DL is given by
		\begin{multline}
	\mathbb{P}\left[\left\{P_{M}\mathbb{E}\left\{\left\|{\bf{h}}_M\right\|^2\right\}X_M^{-\alpha_M}>P_{F}\mathbb{E}\left\{\left\|{\bf{h}}_F\right\|^2\right\}X_F^{-\alpha_F}\right\} \right. \bigcap  \\ \left. \left\{B_M  \mathbb{E}\left\{\left\|{\bf{h}}_M\right\|^2\right\}X_M^{-\alpha_M}> B_F \mathbb{E}\left\{\left\|{\bf{h}}_F\right\|^2\right\}X_F^{-\alpha_F}\right\} \right] \nonumber
	\label{case1}
	\end{multline}
where the expectation is over the channel fading. The $\mathbb{E}\left\{\left\|{\bf{h}}_M\right\|^2\right\}=N_M$, $\mathbb{E}\left\{\left|{h}_F\right|^2\right\}=N_F$, where $N_M$, and $N_F$ are the array gains and represent the number of antennas at a MBS and FBS respectively \cite{bruno_book}. $B_F$ and $B_M$ are bias factors toward femto-tier and  macro-tier respectively. %When $B>1$ then users at the boundary of the macro cell are offloaded to the femto cell, whereas $B<1$ offload users from the femto cell to the macro cell.
The above equation can be equivalently written as
\begin{multline}
	\mathbb{P}\left[\left\{P_{M}N_M X_M^{-\alpha_M}>P_{F} N_F X_F^{-\alpha_F}\right\}\bigcap \right. \\ \left. \left\{B_M N_M  X_M^{-\alpha_M}> B_F N_F  X_F^{-\alpha_F}\right\}\right]. \nonumber
		\end{multline}
We know that $P_F < P_M$ and when $\frac{B_F}{B_M} \geq \frac{P_F}{P_M}$, it can be easily observed that the common region in the above equation is $N_M X_M^{-\alpha_M}> \frac{B_F}{B_M} N_F X_F ^{-\alpha_F}$, or equivalently $X_F > \left(\frac{B_F N_F}{B_M N_M}\right)^{\left(1/\alpha_F\right)} X_M^{\alpha_M/\alpha_F} $. Similarly, when $\frac{B_F}{B_M} < \frac{P_F}{P_M}$ then the common region is $X_F > \left(\frac{P_F N_F}{P_M N_M}\right)^{\left(1/\alpha_F\right)} X_M^{\alpha_M/\alpha_F} $ and the probability is calculated as
\begin{equation}
	\mathbb{P}\left(case 1\right)\hspace{-0.2em}=\hspace{-0.2em}\mathbb{P}\left(X_F > a \right)\hspace{-0.2em} =\hspace{-0.3em} \int_{0}^{\infty}\hspace{-0.5em}\left(1-F_{X_F}\left(a \right)\right)\hspace{-0.2em}f_{X_M}\left(X_M\right)\hspace{-0.2em}\mathrm{d_{X_M}}\hspace{-0.1em}, \nonumber
	\label{}
	\end{equation}
	where $a=\Upsilon_1^{1/\alpha_F} X_M^{\alpha_M/\alpha_F}$, while for  $\frac{B_F}{B_M} \geq \frac{P_F}{P_M}$, $\Upsilon_1=\frac{B_F N_F}{B_M N_M}$ and for $\frac{B_F}{B_M} < \frac{P_F}{P_M}$, $\Upsilon_1=\frac{P_F N_F}{P_M N_M}$. Using the null probability of 2D PPP, $F_{X_F}\left(X_M\right)=1-e^{-\pi\lambda_F X_M^2}$,  $f_{X_M}\left(X_M\right)=2 \pi \lambda_M X_M e^{-\pi\lambda_M X_M^2}$ and evaluating the integral we obtain  \eqref{prob_case1}.  
%\section{}
%\textit{Proof of Lemma 2:} The probability that a user is connected to a macro BS in DL and femto BS in the UL while using UL biasing is given by
%\begin{equation}
	%\mathbb{P}\left[\left\{P_{M} N_M X_M^{-\alpha_M} >   P_{F} N_F X_F^{-\alpha_F} \right\} \bigcap  \left\{B_M N_M X_M^{-\alpha_M}\leq B_F N_F X_F^{-\alpha_F}\right\}\right] \nonumber
	%\label{caseb2}
	%\end{equation}
%where	it can be easily noticed that the common event is $\frac{P_F N_F }{P_M N_M}X_F^{-\alpha_F}<X_M^{-\alpha_M}<\frac{B_F N_F }{B_M N_M} X_F^{-\alpha_F}$, and after some simplification its probability is given by	
%\begin{multline}
	%\mathbb{P}\left(case 2\right)=\mathbb{P}\left(a \leq X_M \leq b\right) = \int_{0}^{\infty}\left(F_{X_M}\left(b\right)-F_{X_M}\left(a\right)\right)f_{X_F}\left(X_F\right)\mathrm{d_{X_F}}, \nonumber
	%\end{multline}
	%where $a=\left(\frac{P_M N_M}{P_F N_F}\right)^{1/\alpha_M}X_F^{\alpha_F/\alpha_M}$ and $b=\left(\frac{B_M N_M}{B_F N_F }\right)^{1/\alpha_M}X_F^{\alpha_F/\alpha_M}$. Putting the value of $F_{X_M}\left(.\right) $ and $f_{X_F}\left(.\right)$ from the null probability of 2D PPP and then evaluating the integral we obtain \eqref{prob_case2}.
%
%
\section{}
\textit{Proof of Lemma 4:} The distance $X_K$ between a typical UE and the tagged BS is a random variable (r.v). The event $X_K>x$ is equivalent to the  event that $X_K > x$ given that a typical user is attached to the $K$th tier (proof follows similar method as in \cite{Han_Geff})
\begin{equation}
\mathbb{P}\left[X_K > x \right]=\mathbb{P}\left[X_K > x | n=K \right] =   \frac{\mathbb{P}\left[X_K > x, n=K\right]}{\mathcal{A}_K}, 
\label{dist_distr}
\end{equation} 
where ${\mathbb{P}\left[n=K\right]} = \mathcal{A}_K$ is the tier association probability given \eqref{assoc_prob_gen}. Let $Pr_K$ and $Pr_J$ be respectively the received power from a typical UE at the nearest $K$th tier and $J$th tier BS then the joint probability $ \mathbb{P}\left[X_K > x, n=K\right]$ is
\begin{multline}
\mathbb{P}\left[X_K > x, n=K\right] = \mathbb{P}\left[X_K > x, Pr_K(X_K)>Pr_j \right]  \\ =  \int_{x}^{\infty}\mathbb{P}\left[ B_K N_K X_K^{-\alpha_K}> B_J N_J X_J^{-\alpha_J}\right]f_{X_K}\left(X_K\right)\mathrm{d_{X_K}} \\
=\int_{x}^{\infty}\mathbb{P}\left[ X_J> \left(\frac{B_J N_J}{B_K N_K}\right)^{1/\alpha_J} X_K^{\alpha_K/\alpha_J}\right]f_{X_K}\left(X_K\right)\mathrm{d_{X_K}}. \nonumber
\end{multline}
From the 2D null probability of PPP we obtain, $\mathbb{P}\left[ X_J> \left(\frac{B_J N_J}{B_K N_K}\right)^{1/\alpha_J} X_K^{\alpha_K/\alpha_J}\right]= \\ \exp\left\{-\pi\lambda_J\left(\frac{B_J N_J}{B_K N_K}\right)^{2/\alpha_J} \left(X_K^{\alpha_K/\alpha_J}\right)^2\right\} $, and $f_{X_K}\left(X_K\right)=2\pi\lambda_K X_K \exp\left\{-\pi \lambda_K X_K^2\right\}$, and plugging in the above equation we get
\begin{multline}
%=\int_{x}^{\infty} \exp\left\{-\pi\lambda_j\left(\frac{B_jN_j}{B_k N_k}\right)^{2/\alpha}X_k^2\right\}2\pi\lambda_k X_k \exp\left\{-\pi \lambda_k X_k^2\right\} \mathrm{d_{X_k}} \\
\mathbb{P}\left[X_K > x, n=K\right]= 2\pi\lambda_K \int_x^{\infty}X_K \times \\ 
\exp\hspace{-0.2em}\left\{\hspace{-0.2em}-\pi \hspace{-0.2em}\left(\hspace{-0.2em}\lambda_K X_K^2\hspace{-0.2em}+\hspace{-0.2em}\lambda_J\hspace{-0.2em}\left(\frac{B_J N_J}{B_K N_K}\right)^{2/\alpha_J}\hspace{-0.2em}\right)\left(\hspace{-0.2em}X_K^{\alpha_K/\alpha_J}\hspace{-0.2em}\right)^2\hspace{-0.2em}\right\}\mathrm{d_{X_K}}. 
\label{joint_prob}
\end{multline}
By plugging \eqref{joint_prob} in \eqref{dist_distr} we get
\begin{multline}
\mathbb{P}\left[X_K > x \right] =\frac{2\pi\lambda_K}{\mathcal{A}_K}\int_x^{\infty}X_K \times  \\  \exp\hspace{-0.2em}\left\{\hspace{-0.3em}-\pi \hspace{-0.3em}   \left(\hspace{-0.3em}\lambda_K X_K^2\hspace{-0.2em}+\hspace{-0.2em}\lambda_J\hspace{-0.2em}\left(\hspace{-0.2em}\frac{B_J N_J}{B_K N_K}\hspace{-0.2em}\right)^{2/\alpha_J}\hspace{-0.2em}\right)\hspace{-0.3em}\left(\hspace{-0.2em}X_K^{\alpha_K/\alpha_J}\hspace{-0.2em}\right)^2\hspace{-0.2em}\right\}\hspace{-0.2em}\mathrm{d_{X_K}},
\end{multline}
which is the complementary cumulative distribution function (CCDF) of $X_K$, while it CDF is $F_{X_K}(x)=1-\mathbb{P}\left[X_K > x \right]$, and probability density function (pdf) is $f_{X_K}\left(x\right) = \frac{d}{dx}F_{X_K}\left(x\right)$, we obtain \eqref{dist_distribution}.
\begin{figure*}[!t]
% ensure that we have normalsize text
\normalsize
% Store the current equation number.
\setcounter{MYtempeqncnt}{\value{equation}}
% Set the equation number to one less than the one
% desired for the first equation here.
% The value here will have to changed if equations
% are added or removed prior to the place these
% equations are referenced in the main text.
\setcounter{equation}{34}
\begin{multline}
\mathcal{L}_{I}\left(s\right) = \mathbb{E}_{I}\left[e^{-sI_K}\right]\overset{a}=\mathbb{E}_{g_i,X_{K_i},D_{K_i}}\left[\exp\left(-s\sum\limits_{i\in\Phi_{K}'\backslash u_0}g_iX_{{K}_i}^{{\alpha}_K\eta}D_{K_i}^{-{\alpha}_K}\right)\right]\mathbb{E}_{g_q,X_{J_q},D_{J_q}}\left[\exp\left(-s \sum\limits_{q\in\Phi_{J}'}g_q X_{J_q}^{\alpha_J \eta}D_{J_q}^{-\alpha_K}\right)\right] \\
\overset{b}=\mathbb{E}_{X_{K_i},D_{K_i}}\hspace{-0.25em}\left[\prod_{i \in \Phi_K' \backslash u_0}\hspace{-0.45em}\mathbb{E}_{g_i}\hspace{-0.15em} \left[\exp\left(-s g_i X_{{K}_i}^{{\alpha}_K\eta}D_{K_i}^{-{\alpha}_K}\right)\right]\hspace{-0.2em}\right] \hspace{-0.25em}\mathbb{E}_{X_{J_q},D_{J_q}}\hspace{-0.15em}\left[\prod_{q \in \Phi_J' }\hspace{-0.1em}\mathbb{E}_{g_q} \hspace{-0.15em}\left[\exp\left(-s g_q X_{J_q}^{\alpha_J \eta}D_{J_q}^{-\alpha_K} \right)\right]\hspace{-0.2em}\right]  \\
\overset{c}= \mathbb{E}_{D_{K_i}}\left[\prod_{i \in \Phi_K' \backslash u_0} \mathbb{E}_{X_{K_i}}\left[\frac{1}{1+s X_{{K}_i}^{{\alpha}_K\eta}D_{K_i}^{-{\alpha}_K}}\right]\right] \mathbb{E}_{D_{J_q}}\left[\prod_{q \in \Phi_J' }  \mathbb{E}_{X_{J_q}}\left[\frac{1}{1+sX_{J_q}^{\alpha_J \eta}D_{J_q}^{-\alpha_K}}\right]\right]\\
\overset{d}= \exp\hspace{-0.1em}\left(\hspace{-0.2em}-2\pi\lambda_K{\int_{X_K}^\infty\hspace{-0.3em}{\left(\hspace{-0.1em}1\hspace{-0.2em}-\hspace{-0.2em}\mathbb{E}_{X_{K_i}}\hspace{-0.2em}\left[\frac{1}{1+s X_{{K}_i}^{{\alpha}_K\eta} u^{-\alpha_K}}\hspace{-0.2em}\right]\right)\hspace{-0.1em}u\mathrm{du}}}\hspace{-0.2em}\right)  
\exp\hspace{-0.2em}\left(\hspace{-0.2em}-2\pi\lambda_J\hspace{-0.2em} {\int_{\left(\frac{N_J B_J X_K^{\alpha_K}}{N_K B_K}\right)^{1/\alpha_J}}^\infty\hspace{-0.2em}{\left(\hspace{-0.2em}1\hspace{-0.2em}-\hspace{-0.2em}\mathbb{E}_{X_{J_q}}\hspace{-0.2em}\left[\hspace{-0.2em}\frac{1}{1+s X_{J_q}^{\alpha_J \eta} v^{-\alpha_K}}\hspace{-0.2em}\right]\hspace{-0.2em}\right)\hspace{-0.2em}v\mathrm{dv}}}\hspace{-0.2em}\right) \\
\overset{e}=\exp\left(-2\pi\lambda_K {\int_{X_K}^\infty{\left(\int_0^u\frac{1}{1+s^{-1} X_{{K}_i}^{{-\alpha}_K\eta} u^{\alpha_K}}f_{X_{K_i}}\left(X_{K_i}\right)\mathrm{d_{X_{K_i}}}\right)u\mathrm{du}}}\right) \times \\ 
\exp\left(-2\pi\lambda_J {\int_{\left(\frac{N_J B_J X_K^{\alpha_K}}{N_K B_K}\right)^{1/\alpha_J}}^\infty{\left(\int_0^{\left(\frac{N_JB_Jv^{\alpha_K}}{N_KB_K}\right)^{1/\alpha_J}}\frac{1}{1+s^{-1} X_{J_q}^{-\alpha_J \eta} v^{\alpha_K}}f_{X_{J_q}}\left(X_{J_q}\right)\mathrm{d_{X_{J_q}}}\right)v\mathrm{dv}}}\right)  \\
\overset{f}=\exp\left(-\pi\lambda_K {\int_0^{X_K}s^{2/\alpha_K} X_{K_i}^{2\eta}{\left(\int_{s^{-2/\alpha_K}X_{K_i}^{2\left(1-\eta\right)}}^\infty   \frac{1}{1+Z_K^{{\alpha_K}/2}} \mathrm{d_{Z_K}}\right)f_{X_{K_i}}\left(X_{K_i}\right)\mathrm{d_{X_{K_i}}}}}\right) \times \\ 
\exp\left(-\pi\lambda_J  \int_{0}^{\zeta^{\frac{\alpha_J+\alpha_K}{\alpha_J^2}} X_K^{\alpha_K^2/\alpha_J^2}} s^{2/\alpha_K} X_{J_q}^{2\alpha_J\eta/\alpha_K}{\left(\int_{\zeta^{-2/\alpha_K}s^{-2/\alpha_K}X_{J_q}^{2\alpha_J\left(1-\eta\right)/\alpha_K}}^\infty\frac{1}{1+Z_J^{{\alpha_K}/2}}\mathrm{d_{Z_J}}\right)f_{X_{J_q}}\left(X_{J_q}\right)\mathrm{d_{X_{J_q}}}}\right)   \\
\overset{g}=\exp\left(\frac{-2\pi\lambda_K s}{\alpha_K-2} {\int_0^{X_K} X_{K_i}^{2-\alpha_K\left(1-\eta\right)}{}_2 \mathrm{F}{}_1\left[1,1-\frac{2}{\alpha_K},2-\frac{2}{\alpha_K}; -s X_{K_i}^{-\alpha_K\left(1-\eta\right)}\right]f_{X_{K_i}}\left(X_{K_i}\right)\mathrm{d_{X_{K_i}}}}\right) \times \\ 
\exp\left(\frac{-2\pi\lambda_J \zeta^{1-2/\alpha_K} s}{\alpha_K-2} \int_{0}^{\zeta^{\frac{\alpha_J+\alpha_K}{\alpha_J^2}} X_K^{\alpha_K^2/\alpha_J^2}} X_{J_q}^{2\alpha_J/\alpha_K-\alpha_J\left(1-\eta\right)}{}_2 \mathrm{F}{}_1\left[1,1-\frac{2}{\alpha_K},2-\frac{2}{\alpha_K}; -s \zeta X_{J_i}^{-\alpha_J\left(1-\eta\right)} \right] f_{X_{J_q}}\left(X_{J_q}\right)\mathrm{d_{X_{J_q}}}\right) \\
\label{laplace_SIR}
\end{multline}
% Restore the current equation number.
\setcounter{equation}{\value{MYtempeqncnt}}
% The IEEE uses as a separator
\hrulefill
% The spacer can be tweaked to stop underfull vboxes.
\vspace*{4pt}
\end{figure*}
\section{}
\textit{Proof of Theorem 1:} We consider multiple antenna BSs and use MRC combining, therefore, the signal channel  follows $\mathrm{Gamma}\left(N_K,1\right)$ , whereas the interfering channel still follows exponential distribution \cite{geff_adhoc}. Let $X_K$ be the distance between a typical UE and its serving $K$th tier BS then the coverage probability $\mathcal{C}_K$ for a given threshold can be written as
\begin{multline}
 \mathcal{C}_K \left(\tau_K\right)\triangleq \mathbb{E}_{X_K}\left[\mathbb{P}\left[\mathrm{SIR}_{X_K}>\tau|X_K\right]\right] \\
=\int_{0}^{\infty}{\mathbb{P}\left[\mathrm{SIR}_{X_K}>\tau_K|X_K\right]f_{X_K}\left(X_K\right)\mathrm{d_{X_K}}} \\ 
=\frac{2\pi \lambda_K}{\mathcal{A}_K}\int_{0}^\infty \mathbb{P}\left[\mathrm{SIR}_{X_K} > \tau_K | X_K \right] X_K \times \\ \exp\hspace{-0.2em}\left\{\hspace{-0.2em}-\pi \hspace{-0.2em} \left(\hspace{-0.2em}\lambda_K X_K^2+\lambda_J\left(\frac{B_J N_J}{B_K N_K}\hspace{-0.2em}\right)^{2/\alpha_J}\hspace{-0.2em}X_K^{2\left(\alpha_K/\alpha_J\right)}\hspace{-0.2em}\right)\right\} \hspace{-0.2em}\mathrm{d_{X_K}},
\label{covm1}
\end{multline}  
%\overset{b}= \frac{2\pi \lambda_M}{A_M}\int_{{X_M}>0}^\infty {\mathbb{P}\left[\frac{P_U h_M X_M^{-\alpha}}{I_{X_M}} > \tau | X_M \right] X_M \exp\left\{-\pi X_M^2\left(\lambda_M+\lambda_F\left(\frac{B_F N_F}{B_M N_M}\right)^{2/\alpha}\right)\right\}\mathrm{d_{X_M}}},
where the last expression follows by plugging $f_{X_K}\left(.\right)$ from \eqref{dist_distribution}. % and $(b)$ from the definition of the SIR. 
For interference limited network the $\mathbb{P}\left[\mathrm{SIR}_{X_K} > \tau_K | X_K \right]$ can be written as
\begin{multline}
\mathbb{P}\left[\mathrm{SIR}_{X_K} > \tau_K | X_K \right] \overset{1}= \\
\mathbb{P}\hspace{-0.3em}\left[\hspace{-0.3em}\frac{\left\|{\bf{h}}_{{K}_0}\right\|^2X_{{K}}^{{\alpha_K\left(\eta-1\right)}}}{\sum\limits_{i\in\Phi_{K}'\backslash u_0}\hspace{-0.2em}\left|\frac{{\bf{h}}_{{K}_0}^H {\bf{h}}_{{K}_i}}{\left\|{\bf{h}}_{{K}_0}\right\|}\right|^2\hspace{-0.4em} X_{{K}_i}^{{\alpha}_K\eta}\hspace{-0.1em}D_{K_i}^{-{\alpha}_K}\hspace{-0.3em}  +\hspace{-0.4em}\sum\limits_{q\in\Phi_{J}'}\hspace{-0.2em}\left|\hspace{-0.2em}\frac{{\bf{h}}_{{K}_0}^H\hspace{-0.3em} {\bf{h}}_{{J}_q}}{\left\|{\bf{h}}_{{K}_0}\right\|}\hspace{-0.2em}\right|^2\hspace{-0.3em}X_{J_q}^{\alpha_J \eta}D_{J_q}^{-\alpha_K} }\hspace{-0.3em} \right] \\ 
= \mathbb{P}\left[ \left\|{\bf{h}}_{{K}_0}\right\|^2  >   s I| X_K \right] \overset{2}= \mathbb{E}_I \left[ \sum_{n=0}^{N_K-1}  s^n I^n e^{-sI}\right]
 \\ \overset{3}=\sum_{n=0}^{N_K-1} \frac{s^n \left(-1\right)^n}{n!}  \frac{\mathrm{d}^n}{\mathrm{d}s^n}\mathcal{L}_I\left(s\right)
\label{mcov_laplace}
\end{multline}
where  $s=\tau_K X_{{K}}^{{\alpha_K\left(1-\eta\right)}}$,  $I=\sum\limits_{i\in\Phi_{K}'\backslash u_0}g_iX_{{K}_i}^{{\alpha}_K\eta}D_{K_i}^{-{\alpha}_K}  +\sum\limits_{q\in\Phi_{J}'}g_q X_{J_q}^{\alpha_J \eta}D_{J_q}^{-\alpha_K}$, $g_i=\left|\frac{{\bf{h}}_{{K}_0}^H {\bf{h}}_{{K}_i}}{\left\|{\bf{h}}_{{K}_0}\right\|}\right|^2$ , and $g_q= \left|\frac{{\bf{h}}_{{K}_0}^H {\bf{h}}_{{J}_q}}{\left\|{\bf{h}}_{{K}_0}\right\|}\right|^2$. $\left(1\right)$ follows due to the definition of $\mathrm{SIR}$, %from \eqref{sinr},
$\left(2\right)$ follows due to $\textbf{h}_{K_0} \sim \mathrm{Gamma}\left(N_K,1\right) $, and $\left(3\right)$ follows due to the Laplace transform identity $\mathcal{L}\left\{I^n e^{-sI}\right\}=\left(-1\right)^n \frac{\mathrm{d}^n}{{\mathrm{d}s}^n}\mathcal{L}_I\left(s\right)$.

Now, we find the Laplace transform $\mathcal{L}_I\left(s\right)$ of the interference, which  can be written as in \eqref{laplace_SIR}, available at the top of this page,
where $\left(a\right)$ follows because the interference is from both the femto-tier and macro-tier's scheduled users, and also they are independent of each other, $\left(b\right)$ is due to the i.i.d assumption of $g_i$ and $g_q$,  and  both $g_i$ and $g_q$ are further independent of point process $\Phi$, $\left(c\right)$ is due to $g_i\sim \exp\left(1\right)$ and $g_q\sim \exp\left(1\right)$, $\left(d\right)$ follows due to the probability generating functional (PGFL) of PPP, which convert an expectation over a point process to an integral $\mathbb{E}\left[\prod_{x\in \Phi}f\left(x\right)\right] = \exp\left(-\lambda \int_{\mathbb{R}^2}\left(1-f\left(x\right)\right)\mathrm{dx}\right)$. It is important to mentioned that in step $\left(d\right)$ the integration limits in both of the integrals are not the same, i.e., the closest  interferer of the serving tier can be at a distance $X_K$ from the typical BS, whereas the closest interferer of the non serving tier should be at a distance $\left(\frac{N_J B_J X_K^{\alpha_K}}{N_K B_K}\right)^{1/\alpha_J}$, as mentioned in Remark \ref{remark2}. In step $\left(e\right)$, we apply the inner expectations, which are required for the power control. Again, it is important to note that the distance distribution of an interfering UE to its serving BS is different from that of the typical UE to the tagged BS and for different tiers the distance distribution of an interfering UE to its serving BS are also different, as mentioned in Remark \ref{remark1}. This difference can be seen by the limits of the inner integral in both exponential.    $\left(f\right)$ follows by changing the integration order, putting $\zeta = \frac{N_J B_J}{N_K B_K}$ and some manipulations while $\left(g\right)$ follows by writing the inner integrals as Gauss hypergeometric functions \cite{book_integral}. We combine the two exponential and plugging it in \eqref{mcov_laplace} and then \eqref{mcov_laplace} into \eqref{covm1}. Thus the proof is completed.

\section*{Acknowledgment}
The authors gratefully acknowledge the excellent feedback provided by the anonymous reviewers.

\end{document}